\documentclass[conference,a4paper,10pt,times]{IEEEtran}

\usepackage{amssymb}
\usepackage{amsmath}
\usepackage{bbm}

\usepackage{hyperref}
\usepackage{url}
\hypersetup{
    colorlinks,
    linkcolor={red},
    citecolor={green},
    urlcolor={magenta}
}

\newcommand{\Real}{{\mathbb R}}

\DeclareMathOperator{\bfx}{{\bf x}}

\usepackage{float}

\usepackage{thmtools,thm-restate}
\usepackage{amsmath,amsfonts,amssymb,amsthm}

\usepackage{todonotes}
\usepackage{soul}

\newtheorem{definition}{Definition}
\newtheorem{thm}{Theorem}
\newtheorem{prop}{Proposition}

\newtheorem{lem}[thm]{Lemma}

\theoremstyle{definition}

\theoremstyle{remark}

\usepackage{algorithm}
\usepackage{algorithmic}
\usepackage{multirow}

\usepackage{caption}

\usepackage[utf8]{inputenc} 
\usepackage[T1]{fontenc}    
\usepackage{hyperref}       
\usepackage{url}            
\usepackage{booktabs}       
\usepackage{amsfonts}       
\usepackage{nicefrac}       
\usepackage{microtype}      
\usepackage{graphicx}
\usepackage{subfig}
\usepackage{xcolor}

\def\shownotes{1}  
\ifnum\shownotes=1
\newcommand{\authnote}[2]{{$\ll$\textsf{\footnotesize #1 notes: #2}$\gg$}}
\else
\newcommand{\authnote}[2]{}
\fi

\newenvironment{blockquote}{%
  \par%
  \medskip
  \leftskip=2em\rightskip=2em%
  \noindent\ignorespaces}{%
  \par\medskip}

\renewcommand{\epsilon}{\varepsilon}

\title{Towards Understanding Limitations of Pixel Discretization
  Against Adversarial Attacks}

\author{
	Jiefeng Chen\,$^{1}$\hspace{6mm}
	Xi Wu\,$^{2}$\hspace{6mm}
	Vaibhav Rastogi\,$^{1}$\hspace{6mm}
	Yingyu Liang\,$^{1}$ \hspace{6mm}
	Somesh Jha\,$^{1}$ \\
	\vspace{1mm}
	$^{1}$\,University of Wisconsin-Madison \hspace{5mm}
	$^{2}$\,Google \\
}

\begin{document}

\maketitle

\begin{abstract}
  Wide adoption of artificial neural networks in various domains has led to an
  increasing interest in defending adversarial attacks against them. Preprocessing
  defense methods such as pixel discretization are particularly attractive in
  practice due to their simplicity, low computational overhead, and
  applicability to various systems. It is observed that such methods work well
  on simple datasets like MNIST, but break on more complicated ones like
  ImageNet under recently proposed strong white-box attacks. To understand the
  conditions for success and potentials for improvement, we study the pixel
  discretization defense method, including more sophisticated variants that take
  into account the properties of the dataset being
  discretized. Our results again show poor resistance against the strong
  attacks. We analyze our results in a theoretical framework and offer strong
  evidence that pixel discretization is unlikely to work on all but the simplest of the
  datasets. Furthermore, our arguments present insights why some other
  preprocessing defenses may be insecure.
	
	%
\end{abstract}


\section{Introduction}
\label{sec:intro}
Deep learning models have been shown to be vulnerable to adversarial perturbations,
which are slight perturbations of natural input data but are classified
differently than the natural input~\cite{szegedy2013intriguing}.
Adversarial robustness, or reducing this vulnerability to adversarial perturbations,
of deep neural networks (DNNs) has received significant attention
in recent years (\cite{dziugaite2016study, guo2017countering, meng2017magnet,
  xu2017feature, buckman2018thermometer,xie2017mitigating,song2017pixeldefend,
  samangouei2018defense}), due to the interest of deploying deep learning systems in various security-sensitive or security-critical applications, such as self-driving cars and identification systems. Among various methods, 
  defense methods that preprocess the input and then pass them to the existing learning systems are particularly attractive in
  practice due to their simplicity, low computational overhead, and
  direct applicability to various neural network systems without interfering with the training.
Various such preprocessing methods have been proposed and achieve good results against some old attacks. Unfortunately, as observed by Athalye et al.~\cite{athalye2018obfuscated},
when applied to typical datasets, these methods fall in front of the stronger Backward Pass Differentiable Approximation attack (BPDA). It is now a general perception that preprocessing methods cannot work against such strong attacks. However, the question about why they fail remain largely open. 

This paper studies
a prototypical preprocessing defense method called \emph{pixel discretization}, taking a first step towards understanding the effectiveness and limitations of such methods.
Pixel discretization works by first constructing a codebook of codewords (also called codes) for pixel values, then discretizing an input vector using these codewords, and
finally feeding the discretized input to a pre-trained base model,
with the hope that adversarial noise can be removed by the discretization.
For example, the color-depth reduction by Xu et al.~\cite{xu2017feature} and the subsequent thermometer encoding technique proposed by \cite{buckman2018thermometer} essentially ``round'' pixels to nearby codewords during preprocessing.
They achieve good results on simple datasets such as MNIST~\cite{mnist} against strong attacks such as Carlini-Wagner~\cite{carlini2017towards},
but later are shown to achieve poor performance against BPDA when applied to datasets of moderate or higher complexity (such as CIFAR-10~\cite{krizhevsky2009learning}).
However, we observe that both of these pixel discretization proposals employ very
simple codes which are \emph{completely independent of the data distribution}.
This thus leads to the following natural and intriguing question:
\vskip 2pt
\begin{blockquote}
  \emph{Are there more sophisticated {\em pixel discretization} techniques,
    especially those that are data-specific, which would be significantly better against
    white-box attacks?}
\end{blockquote}
\vskip 2pt
This paper studies the effectiveness of pixel discretization techniques, including data-specific ones, through extensive theoretical and empirical analyses. Our main results suggest a strong negative answer for the above question: if the underlying base model has poor adversarial robustness, pixel discretization is unlikely to improve robustness except for very simple datasets. On one hand, we obtain strong performance on simple datasets like MNIST, identify conditions under which the method can provably work, and propose a simple certificate that can be computed to lower bound the performance against any attacks within a budget.
On the other hand, poor performance is observed on more complicated datasets like CIFAR-10.
We then provide detailed analysis on why the method fails on such data.

A primary tenet of our study is to view discretization
by codewords as constructing a clustering of pixels, where pixels in the same cluster
gets assigned to the same codeword. Our key observation is a stringent tradeoff between
accuracy and robustness: To allow only a marginal reduction in accuracy,
we would like to keep our clusters of pixels small so that only
close enough pixels map to cluster centers. However, on complex data this requires many small,
closely-packed clusters whose boundary pixels can be easily perturbed to pixels
in other clusters, implying that robustness will not be high enough if the underlying
base model has poor adversarial robustness.

As a result, our theoretical and empirical analyses hinge on the separability of pixels:
If the pixels form well-separated clusters, we can easily map them to cluster
centroids and achieve a discretization where pixel perturbation does not change
the clusters they belong to, implying adversarial robustness. Based on this
understanding, we provide a framework for formally certifying adversarial
robustness. While our theory is able to provide certificates of robustness for
simple datasets such as MNIST, for more complex datasets like CIFAR-10, we show
by various means that pixel discretization is an inadequate defense against
adversarial examples. More specifically, complex datasets do not form well-separated clusters, so pixel discretization is unlikely to
improve adversarial robustness.

Many of our arguments are general and may be used to study other preprocessing
techniques as well. For instance, if a defense is not performed at the pixel level
but at the level of patches of pixels, our arguments will extend to
that defense as well.

In summary, this paper makes the following contributions:
\begin{itemize}
  \item We develop more sophisticated pixel discretization techniques that take
    into account the structure of the data points for discretization. Our experiments show that even these sophisticated defenses do not
    perform well against BPDA attacks.
  \item We develop a theoretical framework for pixel discretization defenses and
    present conditions under which they are robust.
  \item We confirm the general expectation that pixel discretization works well
    on simple datasets like MNIST. In fact, we are able to formally certify that
    pixel discretization on such datasets is robust against any adversarial
    attack. This also shows that gradient masking preprocessing defenses are
    possible for some simple distributions.
  \item We further analyze complex datasets like CIFAR-10 to argue that it is
    unlikely to achieve robustness with pixel discretization on these datasets.
\end{itemize}


The rest of this paper is organized as follows. Section~\ref{sec:background}
presents the necessary background. We then present our experiments with
color-depth reduction defense in Section~\ref{sec:colordepth}. We then explore
data-specific pixel discretization defenses in Section~\ref{sec:dataspecific}.
In Section~\ref{sec:analysis}, we present our theoretical framework for
analyzing the results as well as present further empirical results on MNIST and CIFAR-10
supporting the analysis. We present a discussion with insights on preprocessing defenses in
general in Section~\ref{sec:comments}. Finally, we present related work in
Section~\ref{sec:related} and conclude in Section~\ref{sec:con}.

The experiment code is available at \url{https://github.com/jfc43/pixel-discretization}.

\section{Preliminaries}
\label{sec:background}

This section presents the relevant background for this work, important
definitions, and the datasets used in our evaluations.

\subsection{Background}
Deep neural network (DNN) models are often vulnerable to adversarial
inputs, which are slight perturbations of ``natural'' test inputs but
lead the model to produce adversary-selected outputs different from the
outputs to the natural inputs.  The susceptibility of neural networks
to adversarial perturbations was first discovered
by~\cite{biggio:2013,szegedy2013intriguing}, and a large body of
research has been devoted to it since then. A long line of recent
research is dedicated to attacks with and defenses against adversarial
perturbations.  Several other types of attacks have been explored in
the context of machine learning. These include {\it training-time}
attacks where an adversary wishes to poison a data set so that a
``bad'' hypothesis is learned by an ML-algorithm; {\it model
  extraction} attacks where an attacker learns a model by making
queries to an available model; and a {\it model inversion} attack
where the attacker recovers the inputs used to train a model.  The
reader is referred to Section~\ref{sec:related} for a very brief
survey of these research works.

In this paper, we focus on {\it test-time} attacks on supervised image
classification problems. The input image $\bfx \in \Real^{H \times W \times C}$(the image has $H$ rows, $W$ columns and $C$ channels). 
$\bfx[i,j,l]$ represents the value of $i_{th}$ row, $j_{th}$ column and $l_{th}$ channel of the image. 
$\bfx[i,j]$ represents the pixel at  $i_{th}$ row and $j_{th}$ column. 
We have a DNN model $F_\theta$ (where
$\theta$ are model parameters, and we simply write $F$ if $\theta$ is
clear from the context) that maps $\Real^{H \times W \times C}$ 
to a set of class labels.  In our context we assume
that the classifier $F$ has been trained without any interference from
the attacker (i.e. no training time attacks).  Given an attack
distance metric $\|\cdot\|$, the goal of an attacker is to craft a
perturbation $\delta$ so that $F(\bfx+\delta) \neq F(\bfx)$, where
$\bfx \sim \mathcal{D}$ and $\|\delta\| \le \epsilon$.  There are
several algorithms for crafting the perturbation $\delta$, but in this
paper we use the recent method by ~\cite{athalye2018obfuscated}.

There have been several directions for defenses to test-time attacks,
but we mainly focus on preprocessing defenses. In such defenses, one
designs a preprocessing function $T: \Real^{H \times W \times C} 
\mapsto \Real^{H \times W \times C}$, and with
a base model $F$, the end-to-end predictions are produced as
$F(T(\cdot))$.  In this context there are three types of attacks: (1)
\emph{Black Box Attack} where the attacker can only get zero-order
information of $F$ (i.e. the outputs of $F$ correspond to the given
inputs), (2) \emph{Gray Box Attack} where the attacker knows $F$, but
is unaware of $T$, and (3) \emph{White Box Attack} where the attacker
knows both $F$ and $T$.  This paper considers white-box attacks since
this is the strongest attack model.

Given a preprocessing technique $T$, the goal of an adversary is thus
to find another image $\mathbf{z}$ such that $\| \mathbf{z} - \bfx \|
\le \epsilon$ and $F(T(\mathbf{z})) \not= F(T(\bfx))$.

\subsection{Backward Pass Differentiable Approximation (BPDA) Attack}
Introduced recently by Athalye et al.~\cite{athalye2018obfuscated},
the {\it backward pass differentiable approximation (BPDA)} attack is
a white-box attack that assumes the knowledge of both the classifier
$F$ and the preprocessor $T$. The BPDA-attack, given an input and a
target label, computes adversarial perturbations by gradient descent
over $F(T(\cdot))$. The main problem that BPDA solves is that many
preprocessing techniques ``obfuscate gradients'' by using a $T$ that
is not differentiable or is random. BPDA addresses this issue by
using $F(T(\cdot))$ for the forward pass but switching to a
``differentiable approximation'' of $T$ when computing gradients for
the backward pass. 

\subsection{Adversarial Training}
We briefly describe the technique of adversarial training. For
detailed account the reader is referred
to~\cite{madry2017towards}. Roughly speaking adversarial training
works as follows: before processing data $\bfx$ (e.g. using it in an
iteration of stochastic-gradient descent (SGD)) we apply an algorithm
for crafting adversarial examples, such as BPDA, and obtain an
adversarial example $\bfx^*$. The learning algorithm uses $\bfx^*$ instead
of $\bfx$.  The theoretical underpinnings of adversarial training can be
found in robust optimization.

\subsection{Metrics and Definitions}
We now provide the necessary definitions that are used in our work.

\begin{definition}[Accuracy]
Accuracy of a classifier $F$ is defined as 
\[
\Pr_{(\bfx,y)\sim \mathcal{D}} [ F(\bfx) = y ]
\]
\end{definition}

\begin{definition}[Local robustness predicate]
  Given an image $\bfx$, the condition for an adversary with budget $\epsilon$  to {\it not succeed}
  is defined as the following predicate:
  \begin{eqnarray*}
    \mathcal{R}_\epsilon (\bfx,y) \equiv \{
    F(T(\mathbf{z})) = y \textrm{~for any~} \|\mathbf{z} - \bfx\| \le \epsilon \}
  \end{eqnarray*}
  We call this predicate the {\em robustness predicate}.
\end{definition}
  
\begin{definition}[Local certificate]
  A predicate $P(\bfx,y)$
  is called a {\it local certificate} iff $P(\bfx,y)$ implies $\mathcal{R}_{\epsilon} (\bfx,y)$.
  In other words, if $P(\bfx,y)$ is true then it provides a proof of the robustness at $\bfx$.
\end{definition}

\begin{definition}[Robustness accuracy]
  The following quantity is called {\it robustness accuracy}
  \begin{align*}
    \begin{split}
    \Pr_{(\bfx,y)\sim \mathcal{D}}[ F(T(\mathbf{z})) = y \textrm{~for any~}
      \|\mathbf{z} - \bfx\| \le \epsilon ] \\
    \; = \;  \Pr_{(\bfx,y)\sim \mathcal{D}}  [ \mathcal{R}_\epsilon (\bfx,y) ]
    \end{split}
  \end{align*}
\end{definition}
Robustness accuracy is the measure of robustness across the entire probability
distribution. This quantity is the accuracy under the strongest attack. However,
we cannot measure it directly in experiments. We instead measure the accuracy
under the attacks we perform as an estimation of robustness accuracy.

In general, any norm can be used for the attack distance metric$\|\cdot\|$ used
in the definitions above. We specifically use the $\ell_{\infty}$ norm
($\|\cdot\|_\infty$) for all discussions and results in this paper.
The $\ell_\infty$ norm is defined as
  \begin{align*}
    \|\mathbf{u}\|_\infty = \max_i |u_i|
  \end{align*}
  where $\mathbf{u}$ is a vector $(u_1,\cdots,u_n)$. Throughout this paper,
  any mention of the norm means $\ell_\infty$ norm.

\subsection{Datasets}

We use a variety of simple and complex datasets to understand the interaction of
pixel discretization and adversarial attacks.

\paragraph{MNIST and Fashion-MNIST}
The MNIST dataset~\cite{mnist} is a large dataset of handwritten digits. Each digit
has 5,000 training images and 1,000 test images. Each image is a 28x28
grayscale with the pixel color-depth of 8 bits.

The MNIST dataset is largely considered a simple dataset, so we also consider a
slightly more complex dataset called Fashion-MNIST~\cite{xiao2017fashion}. Images in this
dataset depicts wearables such as shirts and boots instead of digits. The image
format, the number of classes, as well as the number of examples are all
identical to MNIST.

\paragraph{CIFAR-10}
The CIFAR-10 dataset~\cite{krizhevsky2009learning} is also a dataset of 32x32 color images with ten classes,
each consisting of 5,000 training images and 1,000 test images. The classes
correspond to dogs, frogs, ships, trucks, etc.

\paragraph{GTSRB}
The German Traffic Sign Recognition Benchmark (GTSRB)~\cite{Stallkamp2012} is a dataset of
color images depicting 43 different traffic signs. The images are not of a fixed
dimensions and have rich background and varying light conditions as would be
expected of photographed images of traffic signs. There are about 39,000
training images and 12,000 test images.

Some of the images in the dataset are very dark, making it difficult to identify
the right features in the images.  This makes adversarial perturbations on them
easier. Therefore, we remove all images that have an average intensity of less
than 50.

\paragraph{ImageNet}
ImageNet~\cite{imagenet_cvpr09} is a complex dataset with over 14 million images and twenty
thousand categories. We use the subset of ImageNet used by the NIPS Adversarial
Attacks \& Defenses Challenge~\cite{Nips2017Defense}, which contains 1000
development images and 5000 test images from 1001 categories. 

\subsection{Experiment Settings}
\paragraph{Pre-trained Models}
For MNIST and CIFAR-10, we use naturally and adversarially pre-trained models from ~\cite{madry2017towards}.
For ImageNet, we use a naturally pre-trained InceptionResNet-V2 model from ~\cite{szegedy2017inception}
and an adversarially pre-trained InceptionResNet-V2 model from ~\cite{tramer2017ensemble}.

\paragraph{Training Hyper-parameters.}
For MNIST and CIFAR-10, we use the same training hyper-parameters as ~\cite{madry2017towards}. To train models on Fashion-MNIST, we use the same hyper-parameters as those we use to train MNIST, except we use $20$ steps PGD attack with $\epsilon=0.1$ when we do adversarial training. And the same hyper-parameters as those used on CIFAR-10 are used to train models on GTSRB. In order to reduce training time, when we naturally (or adversarially) retrain models, we use naturally (or adversarially) pre-trained model to initialize model parameters and train for 10000 epochs.

\paragraph{Attack Methods}
If gradient approximation is needed, we use the BPDA attack, otherwise we use the PGD attack.
For MNIST, we set $\epsilon=0.3$ and use $100$ steps for the attack. For Fashion-MNIST, we set $\epsilon=0.1$ and use $100$ steps for the attack.
For CIFAR-10, when we use naturally trained model as base model, we set $\epsilon=2$ and use 40 steps attack; when we use adversarially trained model, we set $\epsilon=8$ and use $40$ steps.
For GTSRB, we set $\epsilon=8$ and use $40$ steps.
For ImageNet, when we use naturally trained model as based model, we set
$\epsilon=1$ and use 1 steps attack; when we use adversarially trained model, we
set $\epsilon=4$ and use $10$ steps for the attack.



\section{robustness of color-depth reduction}
\label{sec:colordepth}
\begin{table*}
  \begin{center}
    \begin{tabular}{  c | c | c | c | c | c | c  }
      \hline
      \multirow{2}{*}{Dataset} & \multirow{2}{*}{Base Model} & \multirow{2}{*}{k} & \multicolumn{2}{c}{Pre-trained} & \multicolumn{2}{|c}{Re-trained} \\ \cline{4-7}
      \multirow{6}{*}{MNIST} & \multirow{4}{*}{\begin{minipage}{0.8in}Naturally Trained Model\end{minipage}} &  & Accuracy  & Robustness & Accuracy & Robustness \\ \hline
                               &    & 2 & 98.76\% & 75.39\% & 99.09\% & 80.91\% \\ \cline{3-7}
                                 &   & 256 & 99.17\% & 0.00\% & \multicolumn{2}{c}{N/A} \\ \cline{2-7}
       & \multirow{2}{*}{\begin{minipage}{0.8in}Adversarially Trained Model\end{minipage}} 
      				& 2 & 98.18\% & 97.32\% & 98.49\% & 92.89\% \\  \cline{3-7}
                              &    & 256 & 98.40\% & 92.72\% & \multicolumn{2}{c}{N/A} \\ \hline \hline
         \multirow{4}{*}{Fashion-MNIST} & \multirow{2}{*}{\begin{minipage}{0.8in}Naturally Trained Model\end{minipage}} 
                                   & 2 & 77.18\% & 38.07\% & 86.44\% & 41.41\% \\ \cline{3-7}
                                 &   & 256 & 91.18\% & 0.00\% & \multicolumn{2}{c}{N/A} \\ \cline{2-7}
       & \multirow{2}{*}{\begin{minipage}{0.8in}Adversarially Trained Model\end{minipage}} 
      				& 2 & 82.68\% & 66.52\% & 86.18\% & 69.99\% \\  \cline{3-7}
                              &    & 256 & 86.26\% & 71.53\% & \multicolumn{2}{c}{N/A} \\ \hline \hline
           \multirow{10}{*}{CIFAR-10} & \multirow{5}{*}{\begin{minipage}{0.8in}Naturally Trained Model\end{minipage}} 
                                      & 2 & 38.30\% & 18.05\% & 80.38\% & 46.44\% \\ \cline{3-7}
                                &    &  8& 83.66\% & 9.12\% & 92.27\% & 28.13\% \\ \cline{3-7}
                                &    & 16& 92.71\% & 6.72\% & 94.22\% & 13.84\% \\ \cline{3-7}
                                &    & 32& 94.40\% & 4.51\% & 94.81\% & 6.10\% \\ \cline{3-7}
                                &    & 256 & 95.01\% & 4.20\% & \multicolumn{2}{c}{N/A} \\ \cline{2-7}
       & \multirow{5}{*}{\begin{minipage}{0.8in}Adversarially Trained Model\end{minipage}} 
      			             & 2 & 74.59\% & 34.30\% & 78.35\% & 38.32\% \\ \cline{3-7}
                                &    &  8& 86.51\% & 46.80\% & 87.24\% & 46.89\% \\ \cline{3-7}
                                &    & 16& 87.09\% & 47.11\% & 87.65\% & 47.18\% \\ \cline{3-7}
                                &    & 32& 87.20\% & 47.20\% & 87.60\% & 47.13\% \\ \cline{3-7}
                                &    & 256 & 87.25\% & 45.50\% & \multicolumn{2}{c}{N/A} \\ \hline \hline    
        \multirow{10}{*}{GTSRB} & \multirow{5}{*}{\begin{minipage}{0.8in}Naturally Trained Model\end{minipage}} 
                                      & 2 & 61.91\% & 32.12\% & 68.70\% & 29.60\% \\ \cline{3-7}
                                &    &  8& 95.06\% & 18.28\% & 95.39\% & 18.47\% \\ \cline{3-7}
                                &    & 16& 97.15\% & 13.88\% & 97.15\% & 14.98\% \\ \cline{3-7}
                                &    & 32& 97.34\% & 9.93\% & 97.43\% & 10.02\% \\ \cline{3-7}
                                &    & 256 & 97.35\% & 7.83\% & \multicolumn{2}{c}{N/A} \\ \cline{2-7}
       & \multirow{5}{*}{\begin{minipage}{0.8in}Adversarially Trained Model\end{minipage}} 
      			             & 2 & 68.43\% & 55.55\% & 70.93\% & 55.30\% \\ \cline{3-7}
                                &    &  8& 93.08\% & 74.46\% & 93.80\% & 74.90\% \\ \cline{3-7}
                                &    & 16& 93.85\% & 75.66\% & 94.96\% & 76.65\% \\ \cline{3-7}
                                &    & 32& 94.02\% & 75.45\% & 94.92\% & 76.20\% \\ \cline{3-7}
                                &    & 256 & 94.11\% & 74.34\% & \multicolumn{2}{c}{N/A} \\ \hline \hline  
         \multirow{10}{*}{ImageNet} & \multirow{5}{*}{\begin{minipage}{0.8in}Naturally Trained Model\end{minipage}} 
                                      & 2 & 50.80\% & 24.40\% & \multicolumn{2}{c}{}  \\ \cline{3-5}
                                &    &  8& 88.10\% & 30.80\% &  \multicolumn{2}{c}{}  \\ \cline{3-5}
                                &    & 16& 92.80\% & 30.60\% &  \multicolumn{2}{c}{}  \\ \cline{3-5}
                                &    & 32& 94.40\% & 27.30\% &  \multicolumn{2}{c}{}  \\ \cline{3-5}
                                &    & 256 & 94.50\% & 27.50\% & \multicolumn{2}{c}{\multirow{2}{*}{N/A}} \\ \cline{2-5}
       & \multirow{5}{*}{\begin{minipage}{0.8in}Adversarially Trained Model\end{minipage}} 
      			             & 2 & 56.20\% & 11.90\% & \multicolumn{2}{c}{} \\ \cline{3-5}
                                &    &  8& 94.00\% & 12.90\% & \multicolumn{2}{c}{} \\ \cline{3-5}
                                &    & 16& 96.10\% & 11.80\% & \multicolumn{2}{c}{} \\ \cline{3-5}
                                &    & 32& 96.40\% & 8.40\% & \multicolumn{2}{c}{} \\ \cline{3-5}
                                &    & 256 & 96.40\% & 5.20\% & \multicolumn{2}{c}{} \\ \hline     
    \end{tabular}
    \caption{Results on MNIST, Fashion-MNIST, CIFAR-10, GTSRB and ImageNet with color-depth reduction. 
      $k$ is the number of bins. $k=256$ means we use all possible codes in the input space, that is we don't discretize images. 
      }
    \label{table:simple-binning-experiments}
  \end{center}
\end{table*}

We begin by discussing a simple pixel discretization technique as described by
Xu et al.~\cite{xu2017feature}, a prototypical example of pixel discretization defenses. They describe their techniques as ``feature
squeezing'', which are preprocessing techniques intended to condense a number of
features into a smaller number of features. The intuition is that by giving an
adversary a smaller space from which to select features, it will make performing
an adversarial attack difficult. One of their feature squeezing techniques is
color-depth reduction. This is a ``simple binning'' approach where, for example, a
256-bin color channel (as represented by 8 bits) may be quantized to a 8-bin color
channel. Color values are essentially "rounded" to their nearest color bins.

On an image $\bfx$, we can do $k$-bin discretization via the following function:  
\begin{align*}
  T(\bfx)[i,j,l] = \frac{\lfloor (k-1)\bfx[i,j,l]+0.5\rfloor}{k-1}
\end{align*}
Where, $i\in [H]$, $j\in [W]$ and $l \in [C]$. Here, $[H] = \{1,2,\cdots,H\}$. 
This assumes that the pixel values of the input image have been scaled to be in
$[0,1]$. 

\noindent\textbf{Evaluation.} 
As mentioned by Athalye et al. in their BPDA attack paper, color-depth reduction
can be attacked by approximating the preprocessing function $T$ with the
identity function. Below we reproduce these results with a more sophisticated
differentiable approximation of $T$ and provide details of the robustness on
different datasets. We will also use this approximation later
in the paper when we discuss more sophisticated data-specific pixel discretization.

In our evaluation of color-depth reduction using the BPDA method, we compute
$F(T(\bfx))$ in the forward pass, while for the backward pass, we replace $T(\bfx)$ by
$g(\bfx)$, which is:
\begin{align*}
g(\bfx)[i,j,l] = \frac{\sum_{t=1}^{k} c_t \cdot e^{-\alpha
    \|\bfx[i,j,l]-c_t\|}}{\sum_{t=1}^{k} e^{-\alpha \|\bfx[i,j,l]-c_t\|}}.
\end{align*}
where $i\in [H]$, $j\in [W]$, $l \in [C]$ and $c_t = \frac{t-1}{k-1}$.
Note that when $\alpha \to \infty$, $g(\bfx)=T(\bfx)$.
For MNIST and Fashion-MNIST, we set  $\alpha=10$. For other datasets, we set $\alpha=0.1$.
To evaluate classifier's robustness without discretization, we use the PGD attack.
We define natural accuracy, or simply accuracy, as the accuracy on clean data.
Similarly, robustness accuracy, which we denote as robustness,
is the accuracy under attack.

\noindent\textbf{Results.} We present results on MNIST, Fashion-MNIST, CIFAR-10, 
GTSRB and ImageNet datasets as shown in Table~\ref{table:simple-binning-experiments}. Our intention is to see if we can achieve sufficient
robustness while maintaining accuracy on these datasets through color-depth
reduction.

For MNIST, as Xu et al. showed, we obtain good accuracy with
just 2 bins. As can be seen from the table, we obtain a substantial improvement in both the
naturally trained model and the adversarially trained model. 
Although Fashion-MNIST and MNIST are widely considered as similar datasets, the results for Fashion-MNIST are not as good as MNIST, 
which means color-depth reduction cannot always work on this kind of datasets.
We also obtain results on CIFAR-10, GTSRB and ImageNet. We present detailed results with different bins. As can be
seen from the table, color-depth reduction does not substantially improve
robustness under whitebox attacks.

The difference in robustness results for MNIST dataset and for
other datasets is due to the complexity of the datasets. In Section~\ref{sec:analysis}, we
will further investigate this issue. We next try data-specific pixel
discretization, which may possibly improve robustness over complex datasets.

\section{Data-specific Pixel Discretization}
\label{sec:dataspecific}

The fact that the simple binning does not work naturally leads to the question whether one can design more sophisticated discretization schemes with better performance. In particular, a reason why simple binning fails is that it does not take into account the properties of the data. Consider the following intentionally simplified setting: each image has just one pixel $x$ taking values in $[0, 1]$, and images of class $0$ lie close to $0.6$ and those of class $1$ lie close to $1$. Using simple binning with codewords $\{0, 1\}$ then fails, as the images will all be discretized to $1$. This example then motivates more sophisticated approaches that takes into account the distribution of the data, so as to
improve robustness over simple binning by color-depth reduction.
Our approach in this section is data-specific: we aim to
discretize pixels in a way that takes the density of pixels in the dataset into
account. We wish to derive a codebook of pixels from the dataset, which is used
by $T$ that replaced each pixel with its nearest codeword under a suitable
distance metric. 

We begin by describing our framework for codebook construction and then present
our experimental results on data-specific discretization.

\subsection{Preprocessing Framework and Codebook Construction}
\label{sec:framework-and-codebook-construction}
At the high level, our framework is very simple and has the following two steps:
\begin{enumerate}
\item At training time, we construct a codebook $\mathcal{C}=\{c_1,c_2,\cdots,c_k\}$ for some $k$,
  where $c_i$ is in the pixel space.
  This codebook $\mathcal{C}$ is then fixed during test time.
\item Given an image $\bfx$ at test time, $T$ replaces each pixel $\bfx[i, j]$ with a codeword
  $c \in \mathcal{C}$ that has minimal distance to $\bfx[i,j]$.
\end{enumerate}

Intuitively, on one hand, we would like to have a set of codes that are far apart from each other, similar to the requirements of error correcting codes in coding theory, such that even after adversarial perturbation the discretized result will not change. On the other hand, we would like to lose little information so that we can keep a high accuracy after discretization. 
More precisely, a set of good codes should satisfy the following properties:
\begin{itemize}
\item {\it Separability}. Pairwise distances of codewords is large enough
  for certain distance metric.
\item {\it Representativity}. There exists a classifier $F$ that has good accuracy on the discretized data
  based on $\mathcal{C}$, as described by the framework above.
\end{itemize}

Therefore, one may want to apply common clustering algorithms, such as $k$-means and $k$-medoids,
to find such separable codes. Note however, that these algorithms do not make a
guarantee of separability. In search for separability, we therefore try both
$k$-medoids and another algorithm that we develop based on density estimation.
We next describe the two algorithms.

\subsubsection{Density Estimation-based Algorithm}
We devise a new algorithm to construct separable codes based on density estimation and greedy selection
on all the pixels $\mathcal{P}$ in the training data. This algorithm is described in Algorithm~\ref{alg:find}.
This algorithm takes as input a set of images $\mathcal{D}$, a kernel function for density estimation,
and number of codes $k$ and a distance parameter $r$. It repeats $k$ times and at each time,
first estimates the densities for all pixel values, then adds the one with highest density to the codebook,
and finally, remove all pixels within $r$ distance of the picked.

\begin{algorithm}[t]
  \caption{\textsc{Deriving Codes via Density Estimation}}
  \label{alg:find}
  \begin{algorithmic}[1]
    \REQUIRE A training dataset $\mathcal{D}$, distance parameter $r$ and number of codewords $k$, a kernel function $\rho(\cdot, \cdot)$.
    \ENSURE a set of codewords $\mathcal{C}=\{c_1,c_2,\cdots,c_k\}$.
    \STATE Let $\mathcal{P}$ denote all the pixels from images in $\mathcal{D}$.
    \FOR{$i=1,2,\cdots,k$}
    \STATE For each pixel value $v$, estimate its (unnormalized) density as
    $h[v] = \sum_{p \in \mathcal{P}} \rho(p, v)$.
    \STATE Set $c_i \leftarrow \arg\max_{v} h[v]$, and
    $\mathcal{P} \leftarrow \mathcal{P}-\{p\in \mathcal{P} : \|p-c_i\| \leq r \}$.
    \ENDFOR
  \end{algorithmic}
\end{algorithm}

\noindent\textit{Instantiation}. There are many possible kernel functions we can use to instantiate
Algorithm~\ref{alg:find}. In this work we use the simplest choice, the identity kernel,
$\rho(p_1, p_2) = 1$ if $p_1 = p_2$ and $0$ otherwise. In that case, the density estimation at line 3
above becomes counting the frequencies of a pixel $v$ in $\mathcal{P}$. The values of $k$ and $r$ can be tuned and we will report results for different choices of these parameters. 

\subsubsection{$k$-Medoids Algorithm} 
The $k$-medoids algorithm is similar to the $k$-means algorithm and aims to
minimize the distance between points inside a cluster to the cluster centroid,
a point in the cluster, which is representative of the entire cluster. The
parameter $k$ is an input to the algorithm (a hyperparameter). The algorithm
works as follows: we initially select $k$ points from the given data points as
the medoids.  Each point is then associated with the closest medoid. The cost of
the association is calculated as the sum of pairwise distance between each point
and its medoid. The algorithm then iteratively continues by
switching a medoid with another data point and checking if this switch of
medoids reduces the cost. See Algorithm~\ref{alg:find_kmedoid} for the details.

Typically, the $\ell_1$ distance, also called the Manhattan distance, is used as
the distance metric, and is also the distance used in our experiments. 
If $\ell_2$ distance (the Euclidean distance) is used, the algorithm reduces to the commonly
named $k$-median algorithm. A third option is $\ell_\infty$ distance. Experimental
results show that using $\ell_2$ and $\ell_\infty$ distances lead to similar
performance as $\ell_1$, so we only report those results for $\ell_1$. 

Finally, one can use the popular $k$-means algorithm for constructing the codewords. However, it is known to be sensitive to outliers, and indeed leads to worse performance in our experiments. So we do not include experimental results using the $k$-means algorithm. 

\begin{algorithm}[t]
  \caption{\textsc{Deriving Codes via $k$-Medoids}}
  \label{alg:find_kmedoid}
  \begin{algorithmic}[1]
    \REQUIRE A training dataset $\mathcal{D}$, number of codewords $k$, number of iterations $T$, a distance function $d(\cdot, \cdot)$.
    \ENSURE a set of codewords $\mathcal{C}=\{c_1,c_2,\cdots,c_k\}$.
    \STATE Let $\mathcal{P}$ denote all the pixels from images in $\mathcal{D}$.
		\STATE For any set of codewords $\mathcal{C}$, define the $k$-medoid cost (w.r.t.\ to the distance function $d$) as 
		$$
		\text{cost}(\mathcal{P}, \mathcal{C}) = \sum_{p\in \mathcal{P}} \min_{c' \in \mathcal{C}} d(p, c').
		$$
		\STATE Randomly pick $k$ pixels as the initial medoids $\mathcal{C}=\{c_1,c_2,\cdots,c_k\}$.
    \FOR{$t=1,2,\cdots,T$}
		\FOR{Each pixel $p \not\in \mathcal{C}$ and each $c \in \mathcal{C}$}
		\STATE Let $\mathcal{C}_{c,p} = \mathcal{C} \setminus \{c\} \cup \{p\}$
		\IF{$\text{cost}(\mathcal{P}, \mathcal{C}_{c,p}) < \text{cost}(\mathcal{P}, \mathcal{C})$}
		\STATE Set $\mathcal{C} \leftarrow \mathcal{C}_{c,p}$
		\ENDIF
		\ENDFOR
    \ENDFOR
  \end{algorithmic}
\end{algorithm}


We now present experimental results using the  algorithms using density estimation and $k$-medoids.

\subsection{Experimental Results}
\label{sec:exp}

\begin{figure*}[th]
  \centering
  \subfloat{\includegraphics[width=0.5\linewidth]{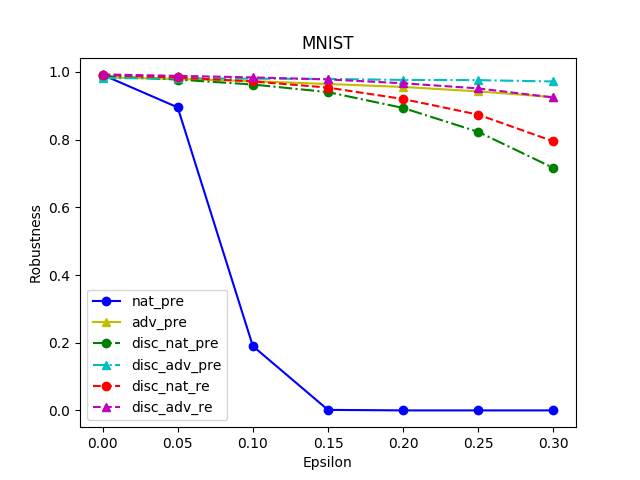} \label{fig:mnist_results}}
  \subfloat{\includegraphics[width=0.5\linewidth]{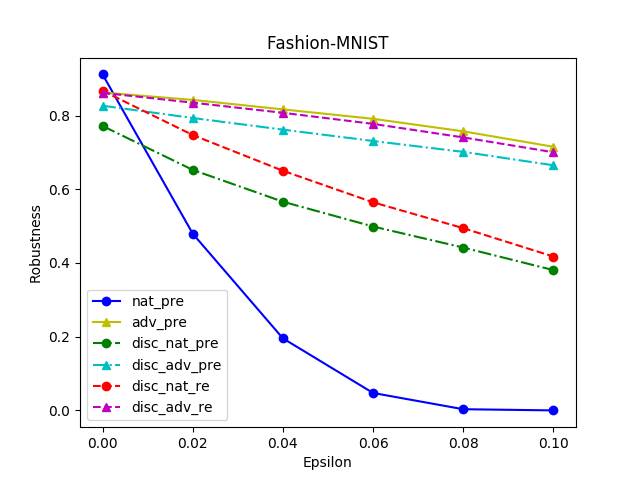} \label{fig:fashion_mnist_results}}
  \caption{
    Results on MNIST and Fashion-MNIST using 2 codes under 100 steps attack.
     The method significantly improves robustness on naturally pre-trained model and adversarially pre-trained model.}
  \label{fig:mnist_fashion_mnist}
\end{figure*}

\begin{figure*}[th]
  \centering
  \subfloat[]{\includegraphics[width=0.5\linewidth]{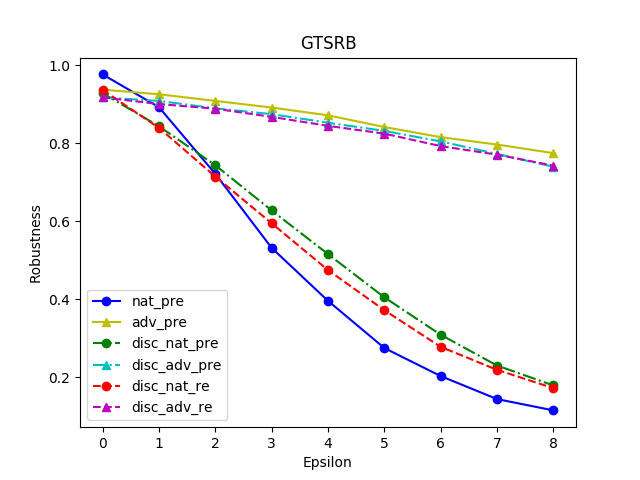} \label{fig:gtsrb_results}}
  \subfloat[]{\includegraphics[width=0.5\linewidth]{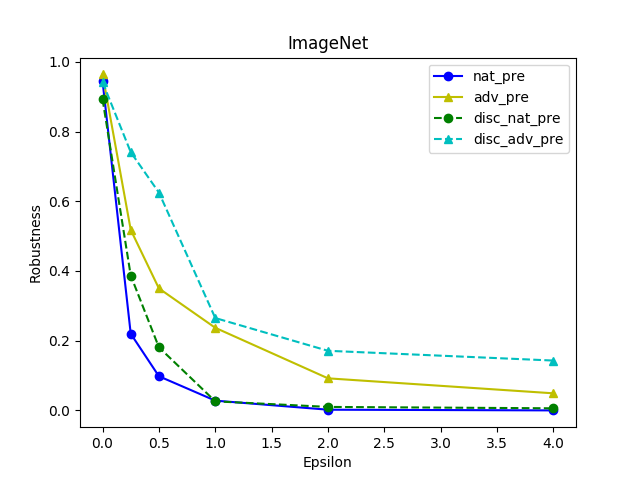} \label{fig:imagenet_results}}
  \caption{
  Results on GTSRB and ImageNet. 
  (a) Results on GTSRB using 50 codes under 40 steps attack. On the
    naturally trained model, our method could improve its robustness. However, there is little improvement in robustness on the adversarially pretrained model. (b) The results on ImageNet using 300 codes under 10 steps attack. The method slightly improves robustness on either naturally pre-trained models or adversarially pre-trained models.}
  \label{fig:gtsrb_imagenet}
\end{figure*}

\begin{figure*}[th]
  \centering
  \subfloat{\includegraphics[width=0.5\linewidth]{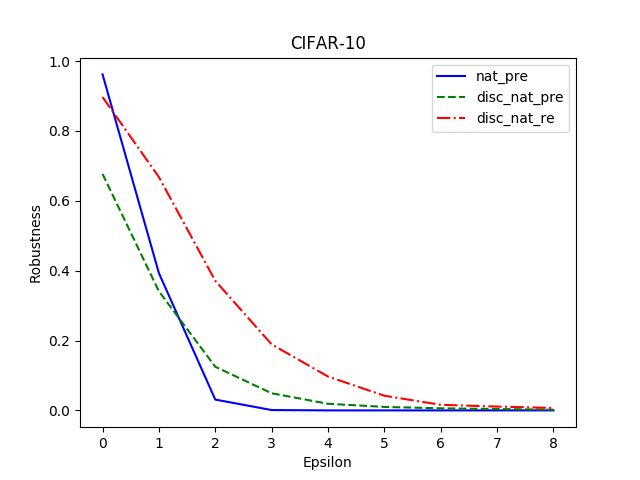} \label{fig:cifar_results}}
	\subfloat{
  \includegraphics[width=0.5\linewidth]{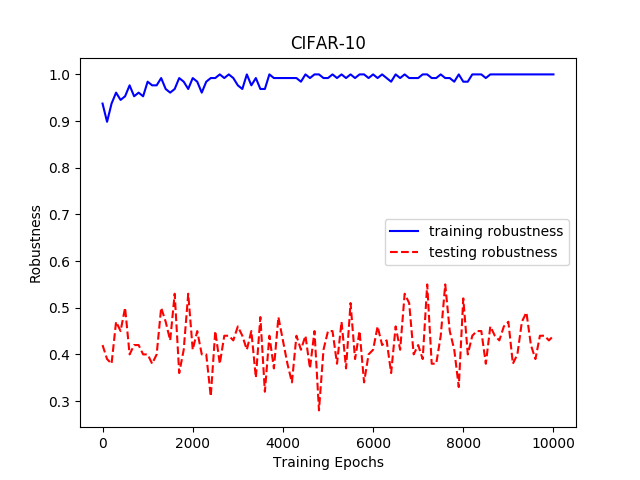} }
  \caption{
    Results on CIFAR-10 using 50 codes under 40 steps attack. 
    (a) On the naturally trained model, the method could improve its robustness. On the adversarially trained model, there is little improvement so we omit the plot. Instead, we plot the adversarial gap underlying this failure.
    (b) Adversarial gap phenomenon on CIFAR-10. 
		When we adversarially retrain the model on discretized images, we can achieve nearly $100\%$ robustness in training but no higher than $50\%$ in test.
    This is referred to as an adversarial generalization gap~\cite{schmidt2018adversarially}. 
		}\label{fig:cifar_gap_results}
\end{figure*}

We now present our evaluation of data-specific discretization. Since
these approaches are more sophisticated than simple color-depth reduction, it is
possible for these approaches to yield robust behavior for a wide range of
datasets. To summarize this
section, however, our key findings are negative: data-specific discretization
techniques do not provide much robustness to complex datasets such as CIFAR-10, GTSRB
and ImageNet. 

Our experimental setup is similar to that of the
previous experiments. We use the same differentiable approximation to $T$, with
$c_i$ this time being the codewords derived from our algorithms and performing operations 
on pixel space.
Our experiments consist of the following three parts:
\begin{itemize}
  \item We would like to evaluate the effectiveness of the specifically designed code construction Algorithm~\ref{alg:find} on different datasets. We would also like to check the variation of robustness as we change the adversarial budget $\epsilon$.
  \item As in the color-depth reduction experiments, we would also like to understand how
    the preprocessing technique provides robustness on different datasets
    quantitatively as we vary the number of codewords.
  \item Finally, we would like to see how the $k$-Medoid algorithm for codebook construction affects the performance. 
\end{itemize}

\noindent\textbf{Effectivenss.} We now evaluate the robustness obtained by pixel
discretization using the codebook constructed by Algorithm~\ref{alg:find} on
different datasets. We consider six settings:
\begin{enumerate}
  \item nat\_pre: no defenses, model naturally trained on original data;
  \item adv\_pre: no defenses, model adversarially trained on original data;
  \item disc\_nat\_pre: discretization defense + model naturally trained on original data;
  \item disc\_adv\_pre: discretization defense + model adversarially trained on original data;
  \item disc\_nat\_re: discretization defense + model naturally trained on preprocessed data;
  \item disc\_adv\_re: discretization defense + model adversarially trained on preprocessed data.
\end{enumerate}
We further vary the adversarial budget $\epsilon$ to give a more complete evaluation of the method.

Figure~\ref{fig:mnist_fashion_mnist}, \ref{fig:gtsrb_imagenet}, and \ref{fig:cifar_gap_results} show the results on the five datasets, with varying $\epsilon$.
We observe that on MNIST, the method improves over the naturally trained model significantly, and further improves the adversarially trained models. 
On Fashion-MNIST, the method also improves over the naturally trained model significantly, but doesn't further improves the adversarially trained models. 
The improvements
with discretization are however still minor for datasets other than MNIST and Fashion-MNIST. 

We also confirm the so-called adversarial generalization gap phenomenon previously reported in~\cite{schmidt2018adversarially}, that is, 
there is a big gap between training and test robustness accuracy. This suggests there is not sufficient data for improving the robustness. 
The right subfigure in Figure~\ref{fig:cifar_gap_results} shows this gap on CIFAR-10.

In general, the results suggest that for complex
datasets (such as CIFAR-10, GTSRB, and ImageNet), it is difficult to achieve robustness via pixel discretization.
This is potentially because the separability and representativity conditions of
codewords cannot be satisfied simultaneously due to the data lacking good
structure and the base model having peculiar properties. 
In Sections~\ref{sec:analysis}, we will study in detail the failure of the method.

\begin{table*}
	\begin{center}
		\begin{tabular}{  c | c | c | c | c | c | c | c  }
			\hline
			\multirow{2}{*}{Dataset} & \multirow{2}{*}{Base Model} & \multirow{2}{*}{k} & \multirow{2}{*}{r} & \multicolumn{2}{c}{Pre-trained} & \multicolumn{2}{|c}{Re-trained} \\ \cline{5-8}
			\multirow{6}{*}{MNIST} & \multirow{4}{*}{\begin{minipage}{0.8in}Naturally Trained Model\end{minipage}} & & & Accuracy  & Robustness & Accuracy & Robustness \\ \hline
			&    & 2 & 0.6 & 98.81\% & 75.35\% & 99.15\% & 80.91\% \\ \cline{3-8}
			&   & 256 & N/A & 99.17\% & 0.00\% & \multicolumn{2}{c}{N/A} \\ \cline{2-8}
			& \multirow{2}{*}{\begin{minipage}{0.8in}Adversarially Trained Model\end{minipage}} 
			& 2  & 0.6 & 98.17\% & 97.24\% & 98.56\% & 92.80\% \\  \cline{3-8}
			&    & 256 & N/A & 98.40\% & 92.72\% & \multicolumn{2}{c}{N/A} \\ \hline \hline
			\multirow{4}{*}{Fashion-MNIST} & \multirow{2}{*}{\begin{minipage}{0.8in}Naturally Trained Model\end{minipage}} 
			& 2 & 0.2 & 77.18\% & 38.07\% & 86.74\% & 41.75\% \\ \cline{3-8}
			&   & 256 & N/A & 91.18\% & 0.00\% & \multicolumn{2}{c}{N/A} \\ \cline{2-8}
			& \multirow{2}{*}{\begin{minipage}{0.8in}Adversarially Trained Model\end{minipage}} 
			& 2 & 0.2 & 82.68\% & 66.52\% & 86.18\% & 70.05\% \\  \cline{3-8}
			&    & 256 & N/A & 86.26\% & 71.53\% & \multicolumn{2}{c}{N/A} \\ \hline \hline
			\multirow{16}{*}{CIFAR-10} & \multirow{8}{*}{\begin{minipage}{0.8in}Naturally Trained Model\end{minipage}} 
			& 2 & 64 & 24.20\% & 15.90\% & 68.50\% & 46.70\% \\ \cline{3-8}
			&    &  5 & 64 & 39.60\% & 18.70\% & 79.90\% & 46.00\% \\ \cline{3-8}
			&    & 10 & 64 & 42.70\% & 16.40\% & 81.60\% & 43.40\% \\ \cline{3-8}
			&    & 50 & 48 & 67.70\% & 12.50\% & 89.70\% & 37.10\% \\ \cline{3-8}
			&    & 100 & 16 & 86.30\% & 7.60\% & 92.80\% & 23.90\% \\ \cline{3-8}
			&    & 300 & 16 & 91.30\% & 7.60\% & 94.40\% & 20.50\% \\ \cline{3-8}
			&    & 500 & 16 & 91.21\% & 9.01\% & 93.53\% & 19.81\% \\ \cline{3-8}
			&    & $256^3$  & N/A & 95.01\% & 4.20\% & \multicolumn{2}{c}{N/A} \\ \cline{2-8}
			& \multirow{8}{*}{\begin{minipage}{0.8in}Adversarially Trained Model\end{minipage}} 
			& 2 & 64 & 63.20\% & 24.40\% & 68.30\% & 33.20\% \\ \cline{3-8}
			&    &  5 & 64 & 76.10\% & 33.40\% & 78.70\% & 40.20\% \\ \cline{3-8}
			&    & 10 & 64 & 78.30\% & 34.70\% & 81.00\% & 43.80\% \\ \cline{3-8}
			&    & 50 & 48 & 84.30\% & 39.50\% & 84.80\% & 43.80\% \\ \cline{3-8}
			&    & 100 & 16 & 86.00\% & 42.50\% & 85.40\% & 44.40\% \\ \cline{3-8}
			&    & 300 & 16 & 87.30\% & 46.90\% & 86.80\% & 46.60\% \\ \cline{3-8}
			&    & 500 & 16 & 86.99\% & 46.77\% & 87.60\% & 46.84\% \\ \cline{3-8}
			&    & $256^3$ & N/A & 87.25\% & 45.50\% & \multicolumn{2}{c}{N/A} \\ \hline \hline    
			\multirow{16}{*}{GTSRB} & \multirow{8}{*}{\begin{minipage}{0.8in}Naturally Trained Model\end{minipage}} 
			& 2 & 64 & 44.02\% & 23.94\% & 55.39\% & 29.46\% \\ \cline{3-8}
			&    &  5 & 64 & 66.65\% & 21.11\% & 80.55\% & 25.42\% \\ \cline{3-8}
			&    & 10 & 64 & 80.51\% & 27.80\% & 84.75\% & 23.87\% \\ \cline{3-8}
			&    & 50 & 48 & 90.01\% & 26.42\% & 91.48\% & 23.12\% \\ \cline{3-8}
			&    & 100 & 16 & 96.62\% & 16.45\% & 96.45\% & 17.91\% \\ \cline{3-8}
			&    & 300 & 16 & 96.80\% & 16.22\% & 96.84\% & 17.08\% \\ \cline{3-8}
			&    & 500 & 16 & 96.82\% & 15.49\% & 96.80\% & 16.25\% \\ \cline{3-8}
			&    & $256^3$  & N/A & 97.35\% & 7.83\% & \multicolumn{2}{c}{N/A} \\ \cline{2-8}
			& \multirow{8}{*}{\begin{minipage}{0.8in}Adversarially Trained Model\end{minipage}} 
			& 2 & 64 & 54.79\% & 43.18\% & 58.53\% & 43.94\% \\ \cline{3-8}
			&    &  5 & 64 & 74.56\% & 53.88\% & 81.43\% & 58.52\% \\ \cline{3-8}
			&    & 10 & 64 & 85.83\% & 66.84\% & 85.93\% & 67.36\% \\ \cline{3-8}
			&    & 50 & 48 & 91.09\% & 74.26\% & 91.34\% & 74.48\% \\ \cline{3-8}
			&    & 100 & 16 & 93.97\% & 74.74\% & 94.72\% & 76.06\% \\ \cline{3-8}
			&    & 300 & 16 & 93.96\% & 75.31\% & 94.78\% & 76.18\% \\ \cline{3-8}
			&    & 500 & 16 & 93.98\% & 75.54\% & 94.73\% & 76.66\% \\ \cline{3-8}
			&    & $256^3$ & N/A & 94.11\% & 74.34\% & \multicolumn{2}{c}{N/A} \\ \hline \hline     
			\multirow{14}{*}{ImageNet} & \multirow{7}{*}{\begin{minipage}{0.8in}Naturally Trained Model\end{minipage}} 
			& 10 & 64 & 53.90\% & 23.00\% & \multicolumn{2}{c}{}  \\ \cline{3-6}
			&    &  50 & 48 & 77.10\% & 28.50\% &  \multicolumn{2}{c}{}  \\ \cline{3-6}
			&    & 100 & 16 & 88.40\% & 32.40\% &  \multicolumn{2}{c}{}  \\ \cline{3-6}
			&    & 200 & 16 & 89.50\% & 32.80\% &  \multicolumn{2}{c}{}  \\ \cline{3-6}
			&    & 300 & 16 & 89.30\% & 34.80\% &  \multicolumn{2}{c}{}  \\ \cline{3-6}
			&    & 500 & 16 & 92.10\% & 34.10\% &  \multicolumn{2}{c}{}  \\ \cline{3-6}
			&    & $256^3$  & N/A & 94.50\% & 27.50\% & \multicolumn{2}{c}{\multirow{2}{*}{N/A}} \\ \cline{2-6}
			& \multirow{7}{*}{\begin{minipage}{0.8in}Adversarially Trained Model\end{minipage}} 
			& 10 & 64 & 62.60\% & 7.50\% & \multicolumn{2}{c}{} \\ \cline{3-6}
			&    &  50 & 48 & 85.30\% & 11.20\% & \multicolumn{2}{c}{} \\ \cline{3-6}
			&    & 100 & 16 & 92.80\% & 12.00\% & \multicolumn{2}{c}{} \\ \cline{3-6}
			&    & 200 & 16 & 94.00\% & 13.30\% & \multicolumn{2}{c}{} \\ \cline{3-6}
			&    & 300 & 16 & 94.20\% & 14.10\% & \multicolumn{2}{c}{} \\ \cline{3-6}
			&    & 500 & 16 & 94.80\% & 16.00\% &  \multicolumn{2}{c}{}  \\ \cline{3-6}
			&    & $256^3$  & N/A & 96.40\% & 5.20\% & \multicolumn{2}{c}{} \\ \hline     
		\end{tabular}
		\caption{Results on MNIST, Fashion-MNIST, CIFAR-10, GTSRB and ImageNet with data-specific pixel discretization. We derive codes via density estimation.  
			$k$ and $r$ are hyper-parameters used to find the codes. $k=256$ in MNIST and Fashion-MNIST and $k=256^3$ in CIFAR-10, GTSRB and ImageNet mean we use all possible codes in the input space, that is we don't discretize images. 
		}
		\label{table:data-specific-discretization-de-experiments}
	\end{center}
\end{table*}

\begin{table*}
  \begin{center}
    \begin{tabular}{  c | c | c | c | c | c | c  }
      \hline
      \multirow{2}{*}{Dataset} & \multirow{2}{*}{Base Model} & \multirow{2}{*}{k} & \multicolumn{2}{c}{Pre-trained} & \multicolumn{2}{|c}{Re-trained} \\ \cline{4-7}
      \multirow{4}{*}{MNIST} & \multirow{3}{*}{Naturally Trained Model} &  & Accuracy  & Robustness & Accuracy & Robustness \\ \hline
                               &    & 2 & 98.81\% & 75.43\% & 98.99\% & 80.66\% \\ \cline{2-7}
       & \multirow{1}{*}{Adversarially Trained Model} 
      				& 2 & 98.17\% & 97.20\% & 98.50\% & 92.20\% \\  \hline \hline
         \multirow{2}{*}{Fashion-MNIST} & \multirow{1}{*}{Naturally Trained Mode} 
                                   & 2 & 82.45\% & 48.44\% & 87.86\% & 51.53\% \\ \cline{2-7}
       & \multirow{1}{*}{Adversarially Trained Model} 
      				& 2 & 84.84\% & 73.63\% & 87.58\% & 74.32\% \\  \hline \hline
           \multirow{2}{*}{CIFAR-10} & \multirow{1}{*}{Naturally Trained Model} 
                                      & 300&  92.73\% & 8.56\% & 93.91\% & 17.08\% \\ \cline{2-7}
       & \multirow{1}{*}{Adversarially Trained Model} 
                                     & 300& 86.97\% & 46.68\% & 87.53\% & 47.11\% \\ \hline \hline
        \multirow{2}{*}{GTSRB} & \multirow{1}{*}{Naturally Trained Model} 
                                      & 300& 97.31\% & 14.61\% & 97.39\% & 15.02\% \\ \cline{2-7}
       & \multirow{1}{*}{Adversarially Trained Model} 
                                      & 300& 93.98\% & 75.56\% & 95.06\% & 76.46\% \\ \hline \hline
         \multirow{2}{*}{ImageNet} & \multirow{1}{*}{Naturally Trained Model} 
                                     & 300 & 92.30\% & 33.80\% &  \multicolumn{2}{c}{\multirow{2}{*}{N/A}}  \\ \cline{2-5}
       & \multirow{1}{*}{Adversarially Trained Model} 
                                       & 300 & 94.90\% & 12.20\% &  \multicolumn{2}{c}{}  \\ \hline                        
    \end{tabular}
    \caption{Results on MNIST, Fashion-MNIST, CIFAR-10, GTSRB and ImageNet with data-specific pixel discretization. We derive codes via k-medoids. 
      $k$ is the number of medoids.}
    \label{table:data-specific-discretization-k-medoids-experiments}
  \end{center}
\end{table*}

\noindent\textbf{Effect of Number of Codewords.}
Given the results, one may wonder whether the performance can be improved by tuning the parameters of the code construction Algorithm~\ref{alg:find}, in particular, the number of codewords $k$.
Here, we study the relationship between number of codewords and accuracy and robustness.
Table~\ref{table:data-specific-discretization-de-experiments} shows the results on MNIST, Fashion-MNIST, CIFAR-10, GTSRB and ImageNet dataset. 
For MNIST and Fashion-MNIST, we only report results using 2 codewords since we can achieve good accuracy and robustness using 2 codewords. 
For other dataset, we report results on a range of $k$. 


From the results, we can know without attacks, high accuracies can be achieved with only a few number of codewords (e.g., 100), especially when models are retrained.
On naturally trained models, with fewer codewords, we gain more robustness. This is because the distances between the codewords are larger and it is harder for the attacker to change the discretization results. 
On adversarially trained models, an increasing number of codewords leads to better robustness. This is different from the naturally trained case, potentially because the the data points are further away from the decision boundary of adversarially trained models, and thus increasing number of codes does not make it easier for the attacker to change the discretization outcome while giving more representativity leading to better results. This also means that discretization does not lead to significantly change in the robustness. This is confirmed by observing that adversarially trained models without discretization can get $45.60\%$ robustness on CIFAR-10. This is further investigated in Section~\ref{sec:analysis}.
  
\noindent\textbf{Using $k$-Medoids.}
Finally, we also try $k$-medoids algorithm for generating the codebook. As mentioned, 
using $\ell_2$ or $\ell_\infty$ distances leads to similar performance as $\ell_1$, so we report only the results for $\ell_1$ in Table~\ref{table:data-specific-discretization-k-medoids-experiments}. As is clear from the table, the results are on par with the results when using the
density estimation algorithm on all datasets except Fashion-MNIST. The results on Fashion-MNIST show that k-medoids can derive better codewords, which supports our argument that discretization should take into account the properties of data. 




\section{Analysis of Results}
\label{sec:analysis}

In this section we address the following question: {\it when does
  pixel discretization work?}  Experimental results show that the
pixel discretization defense methods can achieve strong performance on
the MNIST data set and alike, but fail on more complex data sets, such
as CIFAR-10 and ImageNet. A better understanding of the conditions
under which the simple framework will succeed can help us identify
scenarios where our technique is applicable, and also provide insights
on how to design better defense methods.

We aim at obtaining better insights into its success by the pursuing
two lines of studies. First, we propose a theoretical model of the
distribution of the pixels in the images and then prove that in this
model the pixel discretization provably works. Second, inspired by our
theoretical analysis, we propose a method to compute a certificate for
the performance of the defense method on given images and a given
adversarial budget. This certificate is the lower bound of the
performance for \emph{any} attacks using a given adversarial
budget. If the certified robust accuracy is high, then it means that
the defense is successful for any attack within the budget, not just
for existing. This is much desired in practice, especially considering
that new and stronger attacks are frequently recently proposed for
DNNs. The certificate, as a lower bound, also allows a rough
estimation of the performance, when combined with the robust accuracy
under the currently available attacks which is an upper bound.

\subsubsection{An idealized model and its analysis}
To garner additional insights, we propose and analyze an idealized
model under which we can precisely analyze when the adversarial
robustness can be improved by pixel discretization method using
codewords constructed by Algorithm~\ref{alg:find}.

In the center of our analysis is a generative model for images, i.e.,
a probabilistic model of the distributions of the images.  Roughly
speaking, in this model, we assumes that there exists some
``ground-truth'' codewords that are well separated, and the pixels of
the images are slight perturbations of these codewords, and thus form
well separated clusters. This idealized model of the images are
directly inspired by the known clustering structure of the pixels in
the MNIST data set, i.e., most pixels from MNIST are either close to
$1$ or close to $0$. Given such a clustering structure, the codeword
construction algorithm (Algorithm~\ref{alg:find}) can find codewords
that are good approximations of the ground-truth codewords, and the
adversarial attacks cannot change the discretization results much.  In
summary, our analysis results suggest the following: Suppose that data
is good in the sense that it can be ``generated'' using some
``ground-truth'' codewords that are sufficiently well separated; then,
as long as we can find a $\nu$-approximation for the ground truth
codewords and we have a base model $F$ that is robust with respect to
the $\nu$-budget, it follows $F(T(\cdot))$ is immune to any
adversarial attack with a small budget $\epsilon \le 5\nu$, thus
providing a boost of adversarial robustness. Or in other words,
pre-processing provides a $5$x boost on adversarial robustness.  We
now present details.

\noindent\textbf{An idealized generative model of images.}  Each image
$\bfx$ is viewed as a $d$-dimensional array (a typical image of width
$W$ and height $H$ can be flattened into an array of dimension $W
\times H$).  Suppose each pixel $\bfx[i]$ is $3$-dimensional vector of
discrete values in $[K]$.  Assume that there is a
set of ground-truth codewords $\mathcal{C}^*=\{c^*_1, \ldots,
c^*_k\}$, where the codewords $c^*_i \in [K]^3$ lies in the same space
as the pixels. Also, assume that the codewords are well separated so
that
$
  \|c^*_i - c^*_j\| \ge \Gamma
$
 for some large $\Gamma$.

Now we specify the generative process of an image.  Each image is
generated in two steps, by first generating a ``skeleton image''
$\mathbf{u}$ where each pixel is a codeword, and then adding noise to
the skeleton. We do not make assumptions about the distribution of the
label $y$.  Formally,
\begin{enumerate}
\item $\mathbf{u}$ is generated from some distribution over $(\mathcal{C}^*)^d$,
  where the marginal distributions satisfy
	$$
	\frac{1}{d}\sum_{j\in [d]} \text{Pr}(u_j = c^*_i) = \frac{1}{k},\forall i \in [k].
	$$
\item $\bfx[i] = \mathbf{u}[i] + \zeta_i$, where $\zeta_i$ takes valid
  discrete values so that $\bfx[i] \in [K]^3$, and
  \begin{align*}
    \|\zeta_i\| \le \Gamma/8, ~\text{and}~ \text{Pr}[\|\zeta_i\| = t] = \alpha \exp(-t^2/\sigma^2)
  \end{align*}
  where $t$ takes valid discrete values, $\sigma$ is a parameter,
  and $\alpha$ is a normalization factor. 
\end{enumerate}
We would like to make a few comments about our generative
model. First, the assumption on the skeleton $\mathbf{u}$ is very
mild, since the only requirement is that the probability of seeing any
codeword is the same $1/k$, i.e., randomly pick a pixel in
$\mathbf{u}$ and it is equal to any codeword with the same
probability. This is to make sure that we have enough pixels coming
from any codeword in the training data, so this condition can be
relaxed to $\frac{1}{d}\sum_{j\in [d]} \text{Pr}(u_j = c^*_i) \ge
\epsilon_c$ for a small $\epsilon_c$. We set $\epsilon_c = 1/k$ for
simplicity of the presentation.

For the second step of the generative process, the assumptions that
the noise $\zeta_i$ takes discrete values and that $\|\zeta_i\| \le
\Gamma/8$ are also for simplicity. The actually needed assumption is
that with high probability, the noise is small compared to $\Gamma$,
the separation between the ground-truth codewords.

\noindent\textbf{Quantifying robustness.}  We now prove our main
theoretical result in the idealized model.

As a first step, we first show that the codewords constructed are
quite close to the ground-truth. Formally, we say that a set of
codewords $\{c_1, \dots, c_k\}$ is a \emph{$\nu$-approximation} of
$\mathcal{C}^*$ if for any $i$, $\| c_i - c_i^*\| \le \nu$. For the
above generative model, one can show that the codewords found by
Algorithm~\ref{alg:find} are $\nu$-approximation of the ground truth
$\mathcal{C}^*$.

\begin{lem}
  \label{lem:ideal}
  Let $N$ denote the number of pixels in the training set. For any
  $\delta > 0$, if $\Gamma > 16\nu$ where $\nu$ is defined as in
  Proposition~\ref{prop:ideal}, then with probability at least
  $1-\delta$, Algorithm~\ref{alg:find} with $r=2\nu$ outputs a set of
  codes $\mathcal{C} = \{c_1,\ldots, c_k\}$ such that $\|c_i - c^*_i\|
  \le \nu, \forall i$.
\end{lem}
\begin{proof}
  Let $x$ denote a pixel in the generated images. We have for any $i$,
  $\text{Pr}[x=c^*_i] \ge \alpha/k$, and for any $c$ such that
  $\|c-c^*_i\| \ge \nu$ for all $i$, $\Pr[x=c] \le \alpha
  \exp(-\nu^2/\sigma^2)$.  By Hoeffding's inequality~\cite[Section
    2.6]{book:concentration}, with probability at least $1-\delta$,
  the empirical estimations satisfy
  \begin{align*}
    \widehat{\Pr}[x=c^*_i] & \ge \alpha/k - \gamma,
		\\
		\widehat{\Pr}[x=c] & \le \alpha \exp(-\nu^2/\sigma^2) + \gamma.
  \end{align*}
  Then $\widehat{\Pr}[x=c^*_i] > \widehat{\Pr}[x=c]$. Note that $r=2\nu$ and $\Gamma>16\nu$, so for each $c^*_i$,
  Algorithm~\ref{alg:find} picks a code from the neighborhood around $c^*_i$ of radius $\nu$ exactly once. This completes the proof.
\end{proof}

Our main result, Proposition~\ref{prop:ideal}, then follows from
Lemma~\ref{lem:ideal}.  To state the result, let us call a
transformation $G$ to be a \emph{$\nu$-code perturbation} of
$\mathcal{C}^*$ if given any skeleton $\mathbf{u}$, $G(\mathbf{u})$
replaces any $c^*_i$ in it with a code $c'_i$ satisfying $\|c'_i -
c^*_i\| \le \nu$.  With this definition we show that, on an image
attacked with adversarial budget $\epsilon\le 5\nu$, our
discretization $T$ will output a $\nu$-code perturbation on
$\mathbf{u}$.  Lemma~\ref{lem:ideal} then leads to the following
proposition (proof is straightforward).

\begin{prop} \label{prop:ideal}
  Assume the idealized generative model, and $\Gamma > 16\nu$ where 
  \begin{align}
    \label{eqn:ideal}
    \nu = \sigma \sqrt{\log \frac{1}{1/k - 2\gamma/\alpha}}, ~~ \gamma = \sqrt{\frac{4}{N} \log\frac{K}{\delta}},
  \end{align}
where $N$ is the number of pixels in the training dataset. Assume $r=2\nu$ in Algorithm~\ref{alg:find}.
Then for any $\epsilon\le 5\nu$,
\begin{align*}
  &  \textnormal{Pr}_{(\bfx,\mathbf{u},y)}[ F(T(\mathbf{z})) = y \textrm{~for any~} \|\mathbf{z} - \bfx\|
  \le \epsilon ] 
	\\
	\ge \ & \min_{G}\textnormal{Pr}_{(\bfx,\mathbf{u},y)}[ F(G(\mathbf{u})) = y]
\end{align*}
with probability at least $1-\delta$, where the minimum is taken over all $\nu$-code perturbation $G$.
\end{prop}

Essentially, this proposition says that one can ``reduce'' defending
against $5\nu$ adversarial attacks to defending against $\nu$-code
perturbations $G$. Therefore, as long as one has a base model $F$ that
is robust to small structured perturbation (i.e., the $\nu$-code
perturbation), then one can defend against \emph{any} $5\nu$
adversarial attacks, which is a significant boost of robustness.
Indeed, we observed that the intuition is consistent with the
experimental results: the method gives better performance using an
adversarially trained $F(\cdot)$, than a naturally trained $F$, on
structured data like MNIST.

This analysis also inspires the following simple approach for
computing a certificate of the robust accuracy for a set of codewords
on given images and a given adversarial budget.

\subsubsection{Certification for discretization defense}
The analysis shows that the method succeeds when the adversarial
attack cannot cause significant changes after discretization, and the
base model is robust to the slight change in the discretized image.
In fact, the intuition is also applicable to general cases beyond the
idealized model. In particular, one can empirically check if such
conditions are met when the images, the base model, and the
adversarial budget are all given.  This observation then gives our
algorithm for computing the certificate for the defense method.

Now we formally derive the certificates for the defense. Given the
codewords $\mathcal{C} = \{c_1, c_2, \ldots, c_k\}$ used in the pixel
discretization, for a pixel $\bfx[i]$, let $c^*(\bfx[i])$ denote its
nearest code in $\mathcal{C}$, and define
\begin{align} \label{eqn:disc_deter_can}
\mathcal{C}(\bfx[i]) = \{c:  \|\bfx[i] - c\| < \|c^*(\bfx[i]) - \bfx[i]\| + 2\epsilon \}
\end{align}
where $\epsilon$ is the adversarial budget.
 
Then after perturbation $\delta$ bounded by $\epsilon$, the distance
between the perturbed pixel $\bfx[i] + \delta$ and $c^*(\bfx[i])$ is
\begin{align*}
\| \bfx[i] + \delta - c^*(\bfx[i]) \| 
& \le \| \bfx[i] - c^*(\bfx[i]) \| + \| \delta \|\\
& \le \| \bfx[i] - c^*(\bfx[i]) \| + \epsilon
\end{align*}
by the triangle inequality.  On the other hand, the distance between
the perturbed pixel $\bfx[i] + \delta$ and any $c\not\in \mathcal{C}(\bfx[i])$
is
\begin{align*}
\| \bfx[i] + \delta - c \| & \ge \| \bfx[i] - c\| - \| \delta \| 
\\
& \ge \| \bfx[i] - c \| - \epsilon.
\end{align*}
By the definition of $\mathcal{C}(\bfx[i])$, we know that 
$$
 \| \bfx[i] - c \| - \epsilon > \| \bfx[i] - c^*(\bfx[i]) \| + \epsilon.
$$ So after perturbation, the pixel $\bfx[i]+\delta$ can only be
 discretized to a code in $\mathcal{C}(\bfx[i])$.

Then all possible outcomes of the discretization after perturbation are
\begin{align} \label{eqn:disc_deter_cer1}
  \mathcal{S}(\bfx) = \{\mathbf{z} = (\mathbf{z}[1], \ldots, \mathbf{z}[d]): \mathbf{z}[i] \in \mathcal{C}(\bfx[i]) \}.
\end{align}
This then leads to the following local certificate for a given image,
and global certificate for the whole distribution.

\noindent\textbf{Local certificate.}  For a data point $(\bfx, y)$, if
for any $\mathbf{z} \in \mathcal{S}(\bfx)$ we have $F(\mathbf{z}) =
y$, then it is guaranteed that $\bfx$ is correctly classified by
$F(T(\cdot))$ to $y$ even under the adversarial attack with
$\epsilon$-budget.  Formally, let $\mathbb{I}_F(\bfx,y)$ be the
indicator that $F(\mathbf{z}) = y$ for all $\mathbf{z} \in
\mathcal{S}(\bfx)$, then
\begin{align} \label{eqn:disc_deter_cer2}
& \mathbb{I}_F(\bfx,y) = \mathbb{I}[F(\mathbf{z}) = y, \forall \mathbf{z} \in \mathcal{S}(\bfx)]
\nonumber \\
\le \ & \mathbb{I}[F(T(\mathbf{z})) = y \textrm{~for any~} \|\mathbf{z} - \bfx\| \le \epsilon],
\end{align}
so $\mathbb{I}_F(\bfx,y)$ is a lower bound on the robust accuracy for
this data point. It serves as a local certificate for $(\bfx, y)$.

\noindent\textbf{Global certificate.}
Define 
$
s = \mathbb{E}_{(\bfx,y) \sim D}[\mathbb{I}_F(\bfx,y)].
$
Then clearly, we have
$$
  s \le \Pr_{(\bfx,y)\sim D}[ F(T(\mathbf{z})) = y \textrm{~for any~} \|\mathbf{z} - \bfx\| \le \epsilon],
$$
so $s$ serves as a lower bound for the robustness accuracy of the defense on the whole data distribution. 

Of course, computing the exact value of $s$ is not feasible, as it
requires access to the true data distribution. Fortunately, this
certificate $s$ can be easily estimated on a validation set of data
points. Even with a medium number of samples we can compute an
estimation $\hat{s}_*$ that is with high probability a close lower
bound of $s$.  Formally, applying a standard concentration bound (in
particular, the Hoeffding's inequality) leads to the following.
\begin{prop} \label{prop:disc_deter_cer}
Let $\hat{s}$ be the fraction of $\mathcal{I}_F(\bfx,y)=1$ on a set
$\mathcal{V}$ of $m$ i.i.d.\ samples from the data distribution.  Then
with probability at least $1-\delta$ over $\mathcal{V}$,
\begin{align*}
  & \Pr_{(\bfx,y)\sim D}[ F(T(\mathbf{z})) = y \textrm{~for any~} \|\mathbf{z} - \bfx\| \le \epsilon]
	\\
	\ge \ & \hat{s}_* :=\left( 1 - \sqrt{\frac{1}{2m} \log \frac{1}{\delta}} \right) \hat{s}.
\end{align*}
\end{prop}

Note that computing the certificate needs enumerating
$\mathcal{S}(\bfx)$ that can be of exponential size.  However, when
the pixels are well clustered, most of them are much closer to their
nearest code than to the others, and thus will not be discretized to a
new code after perturbation, i.e., $|\mathcal{C}(\bfx[i])|=1$. Then
$\mathcal{S}(\bfx)$ is of small size and the certificate is easy to
compute.  This is indeed the case on the MNIST data, which allows us
to compute the estimated certificate $\hat{s}$ in our experiments.

\begin{table}
    \centering
    \begin{tabular}{c|c|c|c|c|c}
      \hline
      $\epsilon$ & $b$ & Unable & Success($\hat{s}$) & Fail & $\hat{s}_*$ \\ \hline
      0.00 & 0 & 0.0\% & 98.1\% & 1.9\%  &  96.61\% \\ \hline
      0.05 & 30 & 0.43\% & 97.43\% & 2.14\% & 95.95\% \\ \hline
      0.10 & 30 & 1.3\% & 96.43\% & 2.27\% & 94.97\% \\ \hline 
      0.15 & 26 & 29.18\% & 69.01\% & 1.81\% & 67.96\% \\ \hline
      0.20 & 25 & 74.38\% & 24.96\% & 0.66\% & 24.58\% \\ \hline 
      0.25 & 25 & 89.48\% & 10.37\% & 0.15\% & 10.21\% \\ \hline 
      0.30 & 25 & 95.67\% & 4.31\% & 0.02\% & 4.24\% \\ \hline   
    \end{tabular}
    \vspace*{4mm}
    \caption{Certificate results on MNIST for different adversarial
      budget $\epsilon$. For computational reasons, we set a threshold
      $b$, and if $|\mathcal{S}(\bfx)|>2^{b}$ where
      $\mathcal{S}(\bfx)$ is defined in Eqn
      (\ref{eqn:disc_deter_cer1}), we report it as an Unable
      case. $\hat{s}$ and $\hat{s}_*$ are defined in
      Proposition~\ref{prop:disc_deter_cer}.}
    \label{table:certificate}
\end{table}

\begin{figure}
    \centering
    \includegraphics[width=0.9\linewidth]{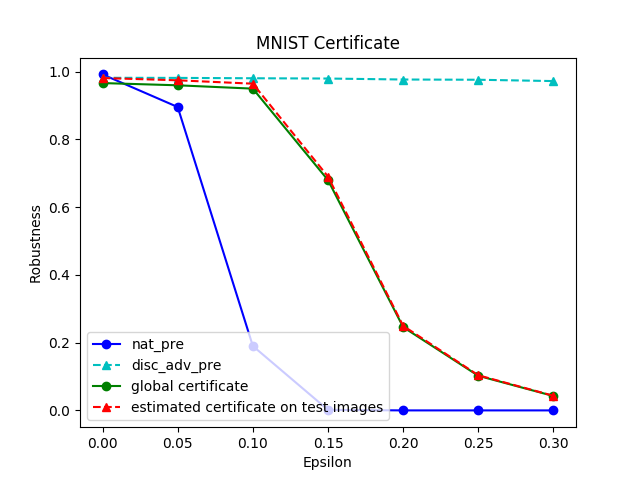}
    \caption{Certificate results on MNIST for the codewords constructed by Algorithm~\ref{alg:find}, compared to the empirical result without defense (nat\_pre), and the empirical result with defense (disc\_adv\_pre). We plot two certificate results: estimated certificate on test images $\hat{s}$, and global certificate $\hat{s}_*$, as defined in Proposition~\ref{prop:disc_deter_cer}. Also see Table~\ref{table:certificate} for the detailed numbers.}
    \label{fig:certificate}
\end{figure}

\subsubsection{Certifying robustness on datasets}

Here we compute the certificate derived above on real world data sets,
including MNIST and others. It is expected that the certificate can be
computed in reasonable time on well structured data like MNIST, while
it may not be efficiently computed or provide no useful lower bound
for other less structured data sets.

\noindent\textbf{MNIST Certificate.}  We compute the estimated
certificate ($\hat{s}$ in in Proposition~\ref{prop:disc_deter_cer}) on
10000 test images of MNIST for the codewords constructed by
Algorithm~\ref{alg:find}. We also compute the global certificate
($\hat{s}_*$ in Proposition~\ref{prop:disc_deter_cer}) for the same
set of codewords, where the failure probability $\delta$ is set to be
0.01.  For computational reasons we set up a threshold $b$ and for
$\bfx$ with $|\mathcal{S}(\bfx)|>2^{b}$ where $\mathcal{S}(\bfx)$ is
defined in (\ref{eqn:disc_deter_cer1}), we report {\sf Unable} and
treat them as failure cases when computing the certificates.

Our results are provided in Table~\ref{table:certificate}, and exact
values of $\hat{s}$ and $\hat{s}_*$ are compared with experimental
results in Figure~\ref{fig:certificate}.  There exists some methods to
compute estimated certificate robustness, like~\cite{katz2017reluplex,
  kolter2017provable, raghunathan2018certified}.  The state-of-art
estimated certificate robustness on MNIST under $\ell_\infty$
perturbations with $\epsilon=0.1$ is
$94.18\%$~\cite{kolter2017provable}, with a fairly sophisticated
method. Our discretization defense, which is much simpler and more
efficient, gives a better estimated certificate robustness of
$96.43\%$.  This demonstrate the effectiveness of this simple
certificate method. It also verifies the analysis for our idealized
generative model of images, providing positive support for the
conclusion that in the presence of well separated clusters of pixels,
the pixel defense method on a good base model can successfully defend
the adversarial attacks.

\subsection{How hard is a dataset to defend?}
\label{sec:hardness}
The discretization defenses against adversarial attacks are effective
only for some datasets (e.g, MNIST). In the last section we analyzed
why on datasets like MNIST our method can succeed. But it is also
interesting and can help improve the method if we can understand why
our techniques fail on other datasets and quantify how far they
are to the conditions for the success.  In this section, we propose
metrics for these purposes and provide empirical results using these
metrics.

Before formally defining these metrics, let us first describe the
concrete context for using these metrics to show the level of hardness
for defense. These include the intuition about the base model, and our
goal for using the metrics.

First, let $F$ denote our base model that is \emph{naturally
  trained}. The fact that they fail miserably under adversarial
attacks with small adversarial budgets suggests that such models have
``peculiar behaviors''. Concretely, for any input feature vector
$\bfx$, we define its robust radius with respect to $F$ and $\ell_p$
norm, $\delta_{F, p}(\bfx)$, as
  \begin{align*}
    \delta_{F, p}(\bfx) = \sup_{\delta}\{ \forall \bfx', \|\bfx' - \bfx\|_p \le \delta, F(\bfx') = F(\bfx) \}.
  \end{align*}
  In other words, $\delta_{F, p}(\bfx)$ is the largest radius (in
  $\ell_p$-norm), under which $F$ still gives consistent predictions
  around $\bfx$. The fact that the base model is vulnerable to small
  adversarial perturbations can then be formalized as the following
  statement:
 \vskip 5pt
 \textbf{Vulnerability of natural models}: \emph{For most points $\bfx$ in the domain,
    $\delta_{F, p}(\bfx)$ is small (``infinitesimal'').}
	\vskip 5pt

Now let $T$ denote the discretization preprocessing algorithm, which
takes as input $\bfx$, and outputs another feature vector $\mathbf{z}
= T(\bfx)$ (of the same dimensionality), and the final classification
is $F(\mathbf{z}) = F(T(\bfx))$. Our goal is the following:
  \vskip 5pt
\textbf{Goal}: \emph{Arguing that it will be very hard to come up with a $T$ that works for
    the naturally trained $F$, except for trivial situations.}
\vskip 5pt

The intuition is that the base model $F$ is peculiar while the discretization is ``regular'' (i.e., the discretized image is a product of the same discretization scheme on each of its pixel). The restriction on the discretization, when facing the peculiarity of the base model, prevents simultaneously getting a good accuracy and getting consistent output in the neighborhood of the input data. 

To get the intuition, consider the following intentionally simplified setting: each image has only one pixel taking values in $[0, 1]$. Divide $[0,1]$ into $m$ many intervals $[(i-1)/m, i/m]$ for $1\le i \le m$, and the images in the $i$-th interval belong to class $0$ if $i$ is even, and belong to class $1$ if $i$ is odd. Now discretize with $k$ codewords. When $m$ is much larger than $k$, no matter where we place the codewords, the accuracy will not be good: if $I_c$ denotes the set of images discretized to the codeword $c$, then we know that many intervals with different classes will fall into the same $I_c$'s, resulting in bad accuracy for any classifiers on top of the discretized images. On the other hand, if $k$ is comparable to $m$, the distance between the codewords will be small and the attacker can easily change the discretized outcome by perturbing the image from one $I_c$ to another $I_{c'}$. Therefore, the fracture decision boundary of the data and the regularity of the discretization combined together can prevent obtaining at the same time good accuracy and unchanged predictions w.r.t.\ adversarial perturbations. 

Now let us provide a more detail argument below.  

\subsubsection{An argument based on equivalence classes}

We think of $T$ as creating equivalence classes, where $\bfx \sim_T
\bfx'$ if $T(\bfx) = T(\bfx')$, that is, they are discretized to the
same output. Now, given an $\ell_p$-norm adversary with adversarial
budget $\varepsilon$, for each input $\bfx$, we can consider the set
of equivalence classes obtained in the ball $N_p(\bfx,
\varepsilon)$. In mathematical terms this is exactly:
\begin{align*}
  \overline{N_p}(\bfx, \varepsilon) = N_p(\bfx, \varepsilon) / \sim_T
\end{align*}
where $\sim_T$ denotes the equivalence relation induced by $T$.

Now, we want to argue the following:
\begin{enumerate}
\item For a naturally trained model $F$, $N_p(\bfx, \varepsilon)$ is
  tiled with infinitesimal balls of different classification labels given by the base model.

\item In order to keep accuracy, we tend to create small equivalence
  classes with $\sim_T$.  In other words, $|\overline{N_p}(\bfx,
  \varepsilon)|$ is large (i.e. there are lot of equivalence classes
  in a neighborhood around $\bfx$). In fact, not only we can assume
  that it is large, but we can assume that \emph{$\overline{N_p}(\bfx,
    \varepsilon)$ subtracting the equivalence class of $\bfx$ has a
    large volume (under uniform measure say).}

\item If the above two hold, then it is now very likely that there is
  a small ball, with a different label, that lies in a different
  equivalence class than $\bfx$. In other words, if we denote
  $\overline{\bfx}$ as the equivalence class of $\bfx$ in $N_p(\bfx,
  \varepsilon)/\sim_T$, then that there is $\bfx'$ such that
  $\overline{\bfx} \neq \overline{\bfx'}$, and $F(\bfx) \neq
  F(\bfx')$, and we can easily find an adversarial example.
\end{enumerate}

\subsubsection{An instantiation for $\ell_\infty$ norm}

We note that for $\ell_\infty$ norm, it is very easy to construct
$\overline{N_\infty}(\bfx, \varepsilon)$.  Specifically, suppose that
the input feature vectors are $d$-dimensional. Then for each ``pixel''
$\bfx[i]$, we can consider the discretization of $\bfx[i]$ under $T$,
i.e., $T(\bfx[i])$ (note that we are abusing the notation and now
consider the effect of $T$ on a single pixel). So for every $i$, we
can consider the equivalence classes created at dimension $i$ for the
interval:
  $C_i = \{ T(z)\ |\ \|z-\bfx[i]\|\leq \varepsilon \}$.
Then all the equivalence classes we can create for an image $\bfx$ are simply:
  $\prod_{i=1}^d C_i$.
Therefore, the $|\overline{N_\infty}(\bfx, \varepsilon)|$ becomes
$\prod_{i=1}^d|C_i|$.  We can thus either use this number as a measure
of how fragmented the ball is under $T$, reflecting how difficult it is to do the defense.

We would like to measure the quantity $\prod_{i=1}^d|C_i|$ for the
datasets we use. In Figures~\ref{fig:mnist-cdf} and
\ref{fig:cifar-cdf}, we plot the CDF (cumulative distribution
function) of $\frac{1}{d}\log_k \prod_{i=1}^d |C_i|$ for both MNIST and CIFAR-10,
where $k$ is the number of codes ($2$ for MNIST and $300$ for CIFAR-10
and $d$ is the number of pixels in one image ($28^2$ for MNIST and
$32^2$ for CIFAR-10).

\begin{figure*}
  \centering
  \subfloat{\includegraphics[width=0.5\linewidth]{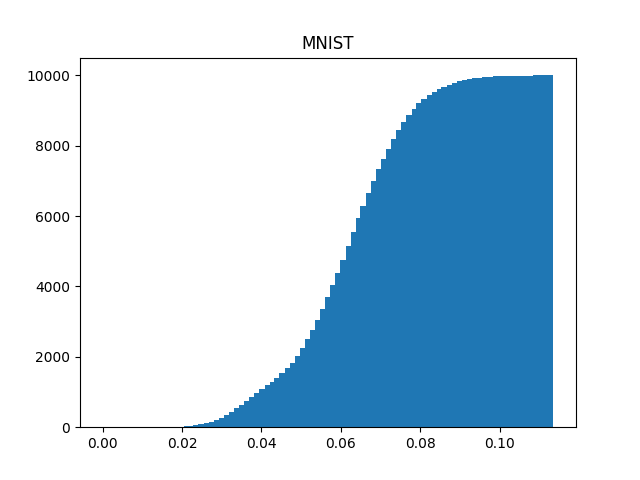}
  \label{fig:mnist-cdf}}
  \subfloat{\includegraphics[width=0.5\linewidth]{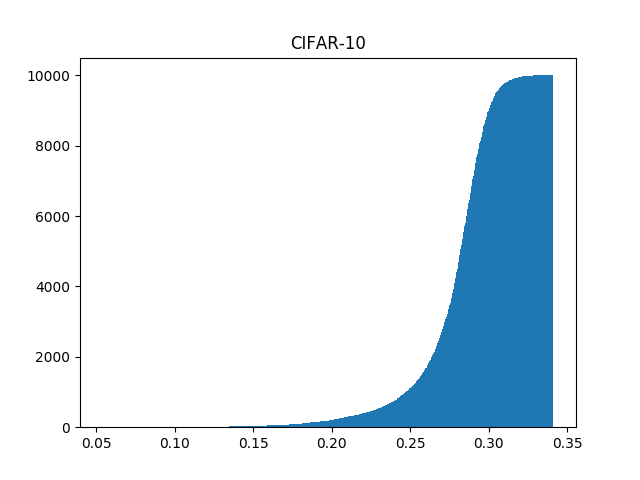}
  \label{fig:cifar-cdf}}
  \caption{
    CDFs measuring the hardness of defending images in MNIST and CIFAR-10.}
  \label{fig:mnist_cifar_cdf}
\end{figure*}

For MNIST, the measure of the median figure is about 0.06, which means
the median $\prod_{i=1}^d|C_i|$ is approximately 47. Out of 784 pixels
of an image, in the median case only 47 can be perturbed to a
different equivalence class.  This shows why it is easy to achieve
robustness on MNIST. For CIFAR-10, the measure for median figure is
about 0.27, which implies that the number of equivalence classes
$\prod_{i=1}^d|C_i|$ is $300^{276}$, which is a huge number and supports
the hypothesis that defending CIFAR-10 will be hard using pixel
discretization techniques.

Furthermore, this equivalence class argument also gives an explanation to the adversarial gap phenomenon (see Figure~\ref{fig:cifar_gap_results}): adversarial training only helps to adjust the labels of the equivalence classes of the perturbations of the training data points, which could be quite different from those classes of the perturbations of the test data points due to the peculiar decision boundary. This then leads to high robust accuracy in training time but low in test time. 

\subsection{How Separable is a Dataset?}
One key requirement for pixel discretization is separability. In a separable
dataset, pixel clusters are far from each other such that perturbation of a pixel
in one cluster cannot make the resultant pixel to move to another cluster. This
is directly related to the equivalence classes argument discussed above.

We visualize MNIST and CIFAR-10 datasets to study if they can have separable
clusters. Figure~\ref{fig:cifar-histogram} presents visualizations of CIFAR-10
pixel neighborhoods. Each axis in the plots corresponds to a color channel.
There are about 4 million distinct pixels in the CIFAR-10 training dataset.
These being too many pixels, we plot only samples of these pixels in this
figure. Figure~\ref{fig:40k-linscale} shows 40,000 pixels on the plot, with each
axis representing a color channel. We assume a colorscale range of $[0,1]$. The
color of each pixel $x$ is given by $|\mathcal{N}_\infty(x, \epsilon)|/\max_x |\mathcal{N}_\infty(x,
\epsilon)|$, where $\mathcal{N}_\infty(x, \epsilon)$ is the $\ell_\infty$ neighborhoods of
radius $\epsilon$ around $x$. For these figures we use $\epsilon=8$. Note that
the maximum neighborhood size is over 1.5 million and there are many pixels with
neighborhood sizes in hundreds or a few thousands and hence most of pixels have
colors close to zero.

To overcome the long-tail distribution of neighborhood sizes, we also do a log
scale plot in Figure~\ref{fig:400k-logscale}, this time using a sample of
400,000 pixels and using  $\log | \mathcal{N}_\infty(x, \epsilon)|/\log \max_x |\mathcal{N}_\infty(x,
\epsilon)|$ to color the pixels on a colorscale of $[0,1]$. As can be seen from
both these plots, \emph{there does not appear to be a clear, separable clustering in
the CIFAR-10 dataset.} Note that the line $x=y=z$, which corresponds to gray
colors on RGB, is lighter-colored, implying large neighborhoods in this area. In
fact, we have verified that our density estimation-based algorithm selects
codewords along this line. Clearly, these codewords are neither very
representative of pixels far from this line nor do they lead to the separability
property for the clusters around them.

\begin{figure*}
  \vskip -3pt
  \centering
  \subfloat[Linear scale]{\includegraphics[width=0.5\linewidth]{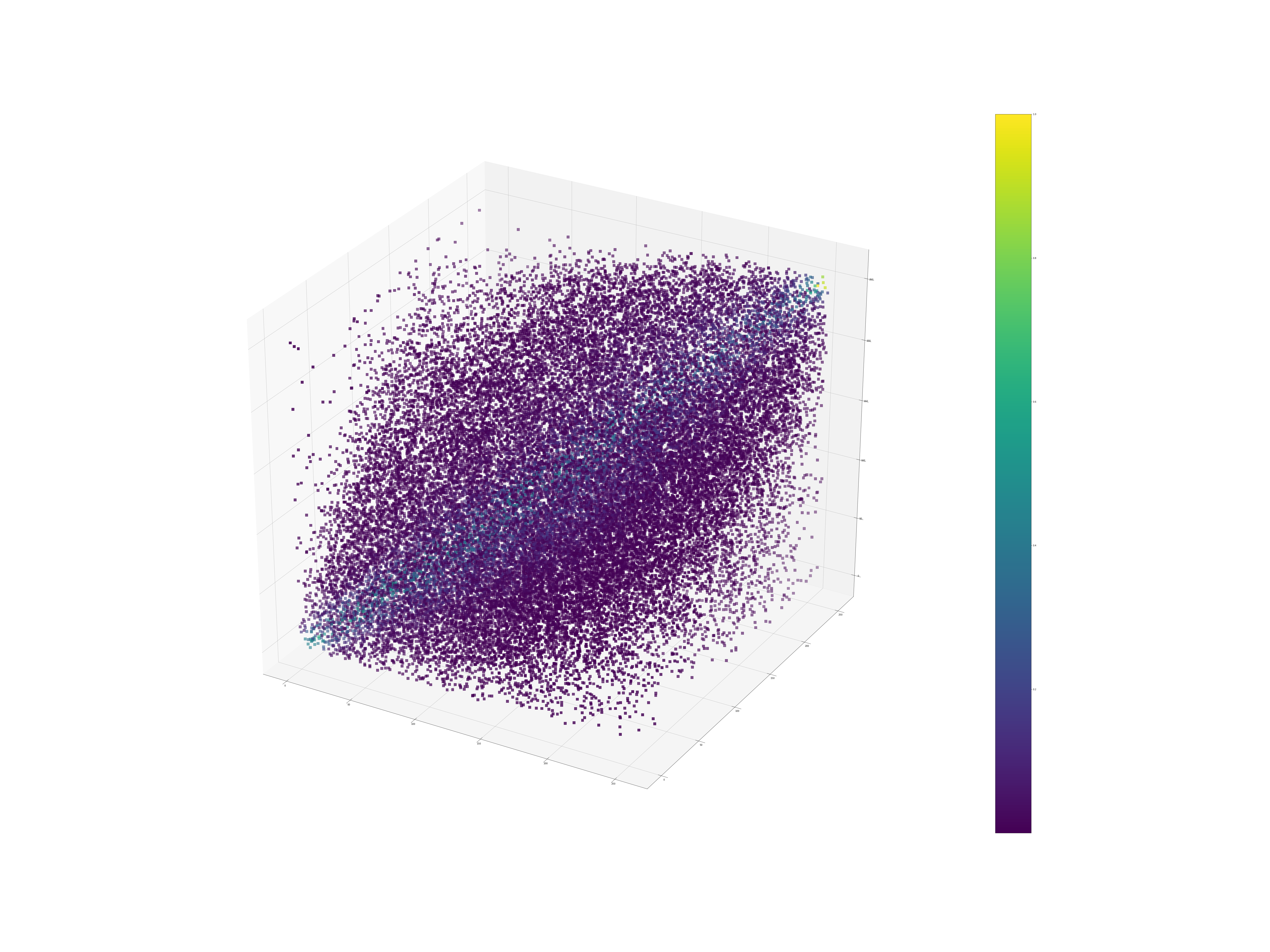}
  \label{fig:40k-linscale}}
  \subfloat[Log scale]{\includegraphics[width=0.5\linewidth]{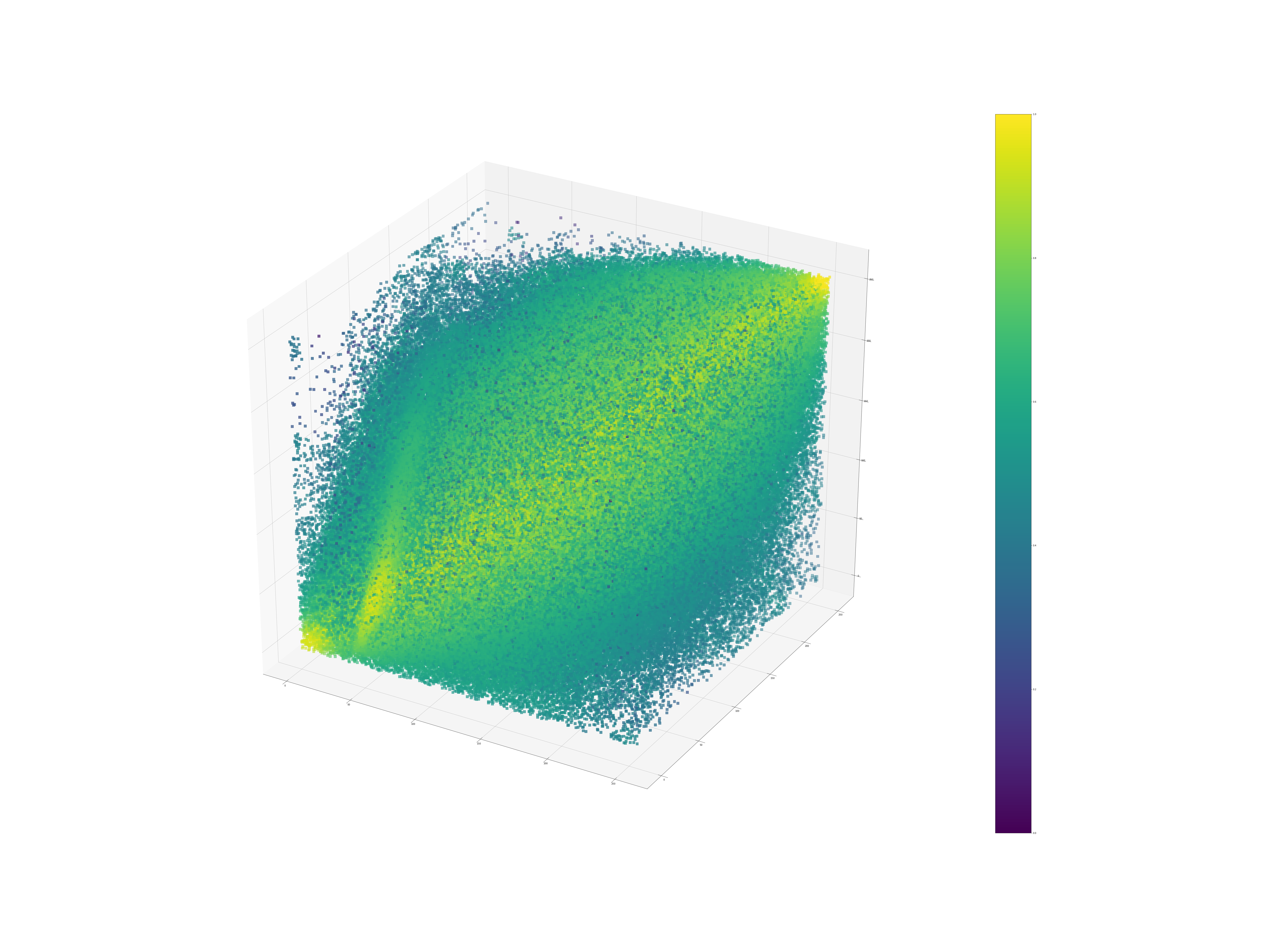}
  \label{fig:400k-logscale}}
  \caption{
    3D visualization of CIFAR-10 pixels colored by their neighborhood size. (a)
    A sample of 40,000 pixels color-coded by their neighborhood sizes in the
    dataset. (b) A sample of 400,000 pixels color-coded by their
    neighborhood sizes on logscale. We normalize the colorscale of the two figures using
    maximum neighborhood size and log of maximum neighborhood size respectively. The maximum
    neighborhood size is 1,567,080.}
  \label{fig:cifar-histogram}
\end{figure*}

\begin{figure}
  \centering
  \includegraphics[width=\linewidth]{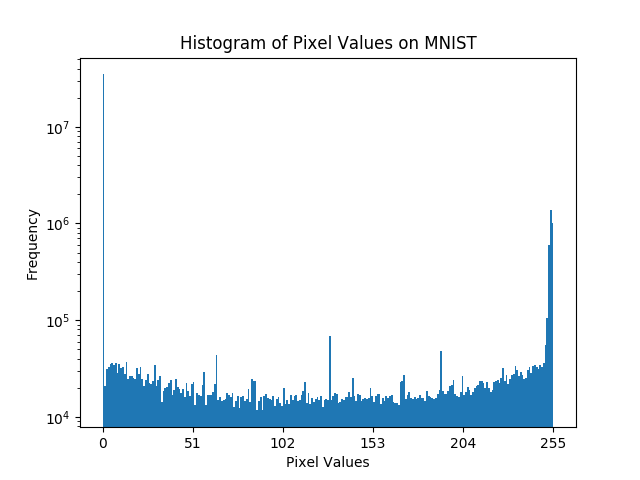}
  \caption{
    A histogram of pixel values in MNIST. Note that $y$-axis is in log scale.
  }
  \label{fig:mnist-histogram}
  \vspace{-3pt}
\end{figure}

We also visualize the MNIST pixels. Figure~\ref{fig:mnist-histogram} presents a
histogram of MNIST pixels. As we can see, most of the pixels are
black or highly white and very few are in between. This leads to \emph{very good
separability} as there are very few pixels that can actually be perturbed to
another cluster or equivalence class.


\section{Comments}\label{sec:comments}

In the previous section, we developed a theory for evaluating pixel
discretization defenses and presented empirical results on CIFAR-10 showing why
it is unlikely for pixel discretization techniques to successfully defend
against adversarial perturbations on this dataset. While our argument until now
was applied on color-depth reduction and the data-specific discretization in Section~\ref{sec:dataspecific}, the
same argument can be applied on other discretizations such as the thermometer
encoding defense by Buckman et al.~\cite{buckman2018thermometer}.

The thermometer encoding defense
encodes each pixel value $\bfx[i,j,l]$ as an $L$-dimensional vector whose $k$-th component
is 1 if $\bfx[i,j,l]$ > $k/L$ and 0 otherwise. Thermometer encoding essentially
rounds a pixel to one of $L$ levels like color-depth reduction but provides a
fancier encoding, which was intended to break gradient descent. Therefore,
the equivalence classes argument that we developed in Section~\ref{sec:hardness}
can be directly applied here. As our results in Figure~\ref{fig:mnist_cifar_cdf}
have shown, CIFAR-10 will be quite difficult to defend with this technique (the
exact statistics on the number of equivalence classes will differ in this
technique and will also depend on $L$).

We also observe that most preprocessing defenses to date have been developed in
an ad hoc manner, with the only evaluation being how well the defense works
against currently-known adversarial attacks. A better lower-bound performance of
a preprocessing defense may be obtained by quantifying how much the
preprocessing technique results in the reduction of equivalence classes for an
image for a given dataset. For pixel discretization, this quantification is
somewhat easy as the equivalence classes at the pixel level can be lifted to
equivalence classes at the image level. For other preprocessing techniques, the
specifics of those techniques will dictate how this quantification is done.

In general, few preprocessing techniques may provide robustness without
significantly affecting accuracy to a naturally trained model. A preprocessor
$T$ is useful only if it increases the robustness radius
(Section~\ref{sec:hardness}) so that $\delta_{F(T(\cdot)),p}(\bfx) =
\delta_{F,p}(T(\bfx)) > \delta_{F,p}(\bfx)$.  However, this requirement may
result in sacrificing some natural examples that have a label different from
$\bfx$ but are nonetheless mapped by $T$ to $\bfx$, resulting in a sacrifice of
accuracy. This is similar to the equivalence class argument in the previous section. Therefore, the design of the preprocessing defense has to be
conservative so as to not sacrifice accuracy too much. A stark example of this 
is in the accuracy-robustness results of pixel discretization on CIFAR-10 and a
naturally-trained model in Table~\ref{table:data-specific-discretization-de-experiments}.
Formalizing this intuition for broader class of preprocessing defense methods is left for future work.

\section{Related Work} 
\label{sec:related}
\paragraph{Adversarial settings in machine learning}
A number of adversarial settings exist in machine learning. The primary ones are
training-time attacks, model inversion, model extraction, and test-time attacks.
Training time attacks poison the training data resulting in learning a wrong
machine learning model. The attacker may skew the classification boundary to
their favor by introducing bad data labeled as good. Data pollution attacks have
been studied long before other attacks became
relevant~\cite{perdisci2006misleading}. In a model inversion attack, the
attacker learns information about data used to train the machine learning
model~\cite{fredrikson2014privacy, fredrikson2015model, wu2016methodology}. A similar but stronger
setting is the membership inference attack where the attacker identifies whether
an individual's information was present in the training
data~\cite{shokri2017membership}. Model extraction attacks attempt to steal a
model simply through blackbox queries~\cite{tramer2016stealing}. All these
settings are different from the setting that this paper focuses on, namely,
small perturbations to natural inputs to get them classified differently
than original inputs. 

\paragraph{Adversarial perturbations}
Adversarial perturbation attacks usually work by starting with a natural example
and solving an optimization problem to derive an adversarial example that has a
different label. A number of techniques with varying settings and efficacy have
been developed~\cite{papernot2016limitations, szegedy2013intriguing,
kurakin2016adversarial, moosavi2016deepfool, carlini2017towards,
madry2017towards}. The BPDA
attack~\cite{athalye2018obfuscated} is helpful when the gradient of a part of
the model is not available so that gradient descent is not possible. The attack
overcomes the problem by using the gradient of a differentiable approximation
of the function whose gradient is not available. We use this attack as the pixel
discretization defense is not differentiable.

\paragraph{Defenses against adversarial perturbations}
Most defenses can be divided into adversarial training defenses and
preprocessing defenses. Currently, the most
successful way to defend against adversarial attacks is adversarial training,
which trains the model on adversarial examples. This training in general imparts
adversarial robustness to the model. The state-of-the-art in this area is Madry
et al.~\cite{madry2017towards}.

Several other defenses fall under the category of preprocessing defenses. These
defenses are model agnostic, not requiring changing the model, and rather simply
transform the input in the home of increasing adversarial robustness. We have
already discussed pixel discretization works of Xu et al.~\cite{xu2017feature}
and Buckman et al.~\cite{buckman2018thermometer}. Other techniques include JPEG
compression~\cite{dziugaite2016study}, total variance
minimization~\cite{guo2017countering}, image quilting~\cite{guo2017countering},
image re-scaling~\cite{xie2017mitigating}, and neural-based
transformations~\cite{song2017pixeldefend, samangouei2018defense,
meng2017magnet}. Total variance minimization randomly selects a small set of
pixels from the input and then constructs an image consistent with this set of
pixels while minimizing a cost function representing the total variance in the
image. The idea is that adversarial perturbations are small and will be lost
when the image is simplified in this way. Similar ideas are behind JPEG
compression. Image quilting on the other hand uses a database of clean image
patches and replaces patches in adversarial images with the nearest clean
patches. Since the clean patches are not adversarial, the hope is that this
approach will undo any adversarial perturbation. Image rescaling and cropping
can change the adversarial perturbations' spatial position, which is important
for their effectiveness. Neural-based approaches train a neural
network to ``reform'' the adversarial examples.

Most of these techniques above have been explicitly broken by the BPDA attack and
others are also believed broken under BPDA attack. We also presented compelling
evidence why pixel discretization cannot work for complex dataset. We believe
that our intuition can be extended to many other classes of preprocessing
techniques and argue that they would not work under strong adversarial settings.

\section{Conclusion}
\label{sec:con}
With all preprocessing defenses developed to date against test-time adversarial
attacks on deep learning models called into question by recently proposed strong
white box attacks, we take a first step towards understanding these defenses
analytically. For this study, we focused on pixel discretization techniques, and,
through a multi-faceted analysis, showed that if the base model has poor
adversarial robustness, pixel discretization by itself is unlikely to improve
robustness on any but the simplest datasets. Our study may pave the way for
a broader understanding of the robustness of preprocessing defenses in general
and guide how to design future preprocessing defenses.


\section{Acknowledgments}
This material is partially supported by Air Force Grant FA9550-18-1-0166, the National Science
Foundation (NSF) Grants CCF-FMitF-1836978, SaTC-Frontiers-1804648
and CCF-1652140 and ARO grant number W911NF-17-1-0405. Any opinions,
findings, conclusions, and recommendations expressed herein are those
of the authors and do not necessarily reflect the views of the funding
agencies. Yingyu Liang would also like to acknowledge that support for this research was provided in part by the Office of the Vice Chancellor for Research and Graduate Education at the University of Wisconsin-Madison with funding from the Wisconsin Alumni Research Foundation.

\bibliographystyle{IEEEtran}
\bibliography{paper}


\end{document}